\documentclass[12pt]{amsart}

\usepackage[all]{xy}

\theoremstyle{plain}
\newtheorem{theorem}{Theorem}[section]
\newtheorem{proposition}[theorem]{Proposition}
\newtheorem{lemma}[theorem]{Lemma}
\newtheorem{corollary}[theorem]{Corollary}

\theoremstyle{definition}
\newtheorem{definition}[theorem]{Definition}
\newtheorem{example}[theorem]{Example}
\newtheorem{remark}[theorem]{Remark}
\newtheorem{notation}[theorem]{Notation}

\makeatletter
\@addtoreset{equation}{section} \makeatother

\begin{document}

\title[Spaces of initial conditions for nonautonomous mappings]{Studies on spaces of initial conditions for nonautonomous mappings of the plane}

\author[T. Mase]{Takafumi Mase} 

\address{Graduate School of Mathematical Sciences, 
the University of Tokyo, 3-8-1 Komaba, Meguro-ku, Tokyo 153-8914, Japan.}


\begin{abstract}
We study nonautonomous mappings of the plane by means of spaces of initial conditions.
First we introduce the notion of a space of initial conditions for nonautonomous systems and we study the basic properties of general equations that have spaces of initial conditions.
Then, we consider the minimization of spaces of initial conditions for nonautonomous systems and we show that if a nonautonomous mapping of the plane with a space of initial conditions, and unbounded degree growth, has zero algebraic entropy, then it must be one of the discrete Painlev\'{e} equations in the Sakai classification.
\end{abstract}

\maketitle


\section{Introduction}\label{sec:intro}

Mappings of the plane are among the main objects of interest in the field of discrete integrable systems.
Such a mapping
\[
	\varphi_n \colon (x_n, y_n) \mapsto (x_{n+1}, y_{n+1})
\]
can be thought of as defining the equation
\[
	(x_{n+1}, y_{n+1}) = \varphi_n (x_n, y_n),
\]
where $x_{n+1}$ and $y_{n+1}$ are functions of $x_n$ and $y_n$ (and $n$).
Hereafter we shall therefore often refer to such a mapping as an equation itself.
A three point mapping, in which $x_{n+1}$ is determined by $x_n$ and $x_{n-1}$, can be transformed to the above form by introducing $y_n = x_{n+1}$.

In this paper, we deal with mappings of the plane that can be rationally solved in the opposite direction.
Such an equation defines a (family of) birational automorphism(s) on $\mathbb{P}^2$ (or on $\mathbb{P}^1 \times \mathbb{P}^1$).

How to detect the integrability of discrete equations has been a major problem in the field of integrable systems for more than a quarter century.

Singularity confinement was first proposed by Grammaticos, Ramani and Papageorgiou \cite{sc} as a discrete analogue of the Painlev\'{e} property in continuous systems.
Where the Painlev\'{e} property requires all movable singularities to be at most poles, singularity confinement requires every singularity (i.e.\,disappearance of information on the initial values) to be confined after a finite number of iterates.
An equation is said to ``enter a singularity'' when loosing information on the initial values, and is said to ``exit from a singularity'' when recovering the lost information.
Singularity confinement is so powerful that many discrete Painlev\'{e} equations have been discovered by deautonomising QRT mappings solely with the help of singularity confinement \cite{dpainleve}.

However, Hietarinta and Viallet presented \cite{hv} an equation that passes the singularity confinement test but which exhibits chaotic behavior.
Their counterexample is
\begin{equation}\label{eq:originalhv}
	x_{n+1} + x_{n-1} = x_n + \frac{a}{x^2_n},
\end{equation}
which is now called the Hietarinta-Viallet equation.
In order to test for integrability more precisely, Bellon and Viallet defined the algebraic entropy \cite{entropy} and showed that the entropy of the above equation is $\log ({(3 + \sqrt{5})}/{2}) > 0$.

\begin{definition}[algebraic entropy \cite{entropy}, dynamical degree]
The limits
\[
	\lim_{n \to \infty} \frac{1}{n} \log \left( \deg \varphi^n \right)
	\text{ and }
	\lim_{n \to \infty} \left( \deg \varphi^n \right)^{1/n},
\]
if they exist, are called the {\em algebraic entropy} and the {\em dynamical degree} of the equation, respectively.
We denote by $\varphi^n$ the $n$-th iterates and by $\deg \varphi^n$ the degree of $\varphi^n$ as a rational function of the initial values in the iteration.
\end{definition}

It is obvious that the entropy coincides with the logarithm of the dynamical degree.

Even in the autonomous case (``autonomous'' meaning that $\varphi$ does not depend on $n$), it is difficult to calculate the exact value of the entropy for a concrete equation.
However, the integrability test based on zero algebraic entropy is empirically accurate.
Hereafter, we shall call an equation with zero algebraic entropy integrable.

\begin{remark}\label{rem:degreeexist}
It is known that in the autonomous case, the entropy exists in nonnegative real number and that is invariant under coordinate changes \cite{entropy}.
However, as in Example~\ref{ex:power}, this does not hold in the nonautonomous case.

\begin{remark}
There are several definitions for the degree of a mapping of the plane.

The degree as a birational automorphism on $\mathbb{P}^2$ (Definition~\ref{app:p2degree}) is the most standard one.
We will mainly use this degree in this paper.

If $\varphi$ is written as
\[
	\varphi(x, y) = \left( \frac{\varphi_{11}(x, y)}{\varphi_{21}(x, y)}, \frac{\varphi_{12}(x, y)}{\varphi_{22}(x, y)} \right),
\]
where $\varphi_{1i}$ and $\varphi_{2i}$ have no common factors for $i = 1, 2$, then the degree of $\varphi$ as a birational automorphism on $\mathbb{P}^1 \times \mathbb{P}^1$ is defined by
\[
	\deg \varphi = \max(\deg \varphi_{11}, \deg \varphi_{12}, \deg \varphi_{21}, \deg \varphi_{22}).
\]
This degree is particularly convenient when we consider three point mappings.

It is known that, while these two degrees are different, their growth as a function of $n$ is the same.
\end{remark}

\begin{example}
Consider the equation
\[
	\varphi(x, y) = \left( \frac{1}{y}, \frac{1}{x} \right).
\]
It immediately follows from the above expression that the degree of $\varphi$ as a birational automorphism on $\mathbb{P}^1 \times \mathbb{P}^1$ is $1$.

On the other hand, $\varphi$ can be written in homogeneous coordinates on $\mathbb{P}^2$ as
\[
	\varphi(z_1 : z_2 : z_3) = (z_1 z_3 : z_2 z_3 : z_1 z_2).
\]
Therefore, the degree of $\varphi$ as a birational automorphism on $\mathbb{P}^2$ is $2$.

Since $\varphi^2 = \operatorname{id}$, the degree growth of $\varphi^n$ is bounded in both cases.
\end{example}

In the nonautonomous case, the degree of the $n$-th iterate of a mapping $\varphi$ is
\[
	\deg \varphi^n = \deg(\varphi_{\ell+n-1} \circ \cdots \circ \varphi_{\ell}),
\]
which in general depends on the starting index $\ell$.
However, we rarely think of $\deg \varphi^n$ (or the entropy) as a function of $\ell$.
It is usual to fix the starting index (for example $\ell = 0$, as in Example~\ref{ex:power}) or only consider these cases where $\deg \varphi^n$ does not depend on $\ell$, for all $n$.
If so, then the algebraic entropy always exists for the same reason as in the autonomous case.
\end{remark}

It has become quite clear that there are many nonintegrable systems that pass the singularity confinement test \cite{hv, lateconfinement, bedfordkim, redemption, redeeming}.
Moreover, most linearizable mappings, which are by definition integrable, do not pass the singularity confinement test \cite{linearizable}.

Besides singularity confinement and algebraic entropy, some other integrability criteria have also been proposed.

Based on Diophantine approximations, Halburd proposed a new integrability criterion, called Diophantine integrability \cite{diophantine}.
This approach is particularly useful when we numerically estimate the value of the entropy.
The coprimeness condition was proposed to reinterpret singularity confinement from an algebraic viewpoint \cite{coprimekdv, coprime, coprimetoda}.
This criterion focuses on the factorization of iterates as rational functions of the initial values and tries to transform the equation to another one with the Laurent property \cite{laurent}.
This method has been recognized as a useful technique to calculate the exact value of the algebraic entropy \cite{extendedhv}.

Since all equations in this paper are (families of) birational automorphisms, geometric methods are useful to analyze them.
The most important and powerful tool is the so-called space of initial conditions, which was first introduced by Okamoto \cite{okamoto} to analyze the continuous Painlev\'{e} equations.

Sakai focused on the close relation that exists between singularity confinement and a space of initial conditions.
Using a special type of algebraic surface, so-called generalized Halphen surfaces, he has classified all discrete Painlev\'{e} equations \cite{sakai}.

Takenawa performed the blow-ups for the Hietarinta-Viallet equation (regularized as an automorphism on a surface) to obtain its space of initial conditions \cite{takenawa1}.
He revealed a correspondence between the singularity pattern and the motion of specific curves, and recalculated the algebraic entropy by computing the maximum eigenvalue of the linear transformation induced on the Picard group.
He also considered blow-ups of nonautonomous systems and showed, by using specific bases first introduced by Sakai, that the degree growth of every discrete Painlev\'{e} equation is at most quadratic \cite{takenawa2}.

Let us start by recalling the close relationship between singularity confinement and the space of initial conditions.

\begin{example}\label{ex:firstexample}
Consider the equation
\begin{equation}\label{eq:qp13point}
	x_{n+1} = \frac{(x_n + 3a) x_{n-1} - 2 a x_n}{x_n - 3 a},
\end{equation}
where $a$ is a nonzero constant.
Although the gauge $x_n = a \widetilde{x}_n$ enables us to take $a = 1$ without loss of generality, this gauge does not work on the nonautonomous version of the equation we will consider in Example~\ref{ex:bupnormal}.
Thus, we do not use the gauge even in the autonomous case.

First let us explain the singularity confinement property on the above equation.
Let $\varepsilon$ be an infinitesimal quantity and assume that while $x_{n-1}$ takes a regular finite value, $x_n$ becomes $3 a + \varepsilon$.
Then we obtain
\begin{align*}
	x_{n+1} &= 6 a (- a + x_{n-1}) \varepsilon^{-1} + o(\varepsilon^{-1}),
	& x_{n+2} &= a + o(1), \\
	x_{n+3} &= 12 a (a - x_{n-1}) \varepsilon^{-1} + o(\varepsilon^{-1}),
	& x_{n+4} &= -a + o(1), \\
	x_{n+5} &= 6 a (- a + x_{n-1}) \varepsilon^{-1} + o(\varepsilon^{-1}),
	& x_{n+6} &= - 3 a + o(1), \\
	x_{n+7} &= - x_{n-1} + o(1),
\end{align*}
where ``$o(\varepsilon^k)$'' is the Landau symbol, i.e.\,$\lim_{\varepsilon \to 0} o(\varepsilon^k)/\varepsilon^k = 0$.
Since the leading order of $x_{n+7}$ is degree $0$ and the leading coefficient again depends on the initial value $x_{n-1}$, we say that this singularity is confined and its pattern is
\begin{equation}\label{eq:ex1scpattern1}
	\{ 3a, \infty^1, a, \infty^1, -a, \infty^1, -3a \}.
\end{equation}
This equation has one more confined singularity pattern
\begin{equation}\label{eq:ex1scpattern2}
	\{ -3a, -a, a, 3a \},
\end{equation}
which starts with $x_n = -3a + \varepsilon$.

Next, we blow up $\mathbb{P}^1 \times \mathbb{P}^1$ to obtain a space of initial conditions.
Equation (\ref{eq:qp13point}) can be written as
\begin{equation}\label{eq:ex1p1p1}
	\varphi \colon \mathbb{P}^1 \times \mathbb{P}^1 \dasharrow \mathbb{P}^1 \times \mathbb{P}^1, \quad
	(x, y) \mapsto
	\left( y, \frac{(y + 3a) x - 2 a y}{y - 3 a} \right).
\end{equation}
Introducing the variables $s = 1/x$ and $t = 1/y$, $\mathbb{P}^1 \times \mathbb{P}^1$ is covered with $4$ copies of $\mathbb{C}^2$ as:
\[
	\mathbb{P}^1 \times \mathbb{P}^1 = (x, y) \cup \left( s, y \right) \cup \left( x, t \right) \cup \left( s, t \right).
\]
Let $X$ be the surface obtained by blowing up $\mathbb{P}^1 \times \mathbb{P}^1$ at the following $9$ points (Figure~\ref{fig:ex1blowup}):
\begin{itemize}
\item
$P^{(1)} \colon \left( x, t \right) = (3a, 0)$,

\item
$P^{(2)} \colon \left( s, y \right) = (0, a)$,

\item
$P^{(3)} \colon \left( x, t \right) = (a, 0)$,

\item
$P^{(4)} \colon \left( s, y \right) = (0, -a)$,

\item
$P^{(5)} \colon \left( x, t \right) = (-a, 0)$,

\item
$P^{(6)} \colon \left( s, y \right) = (0, -3a)$,

\item
$Q^{(1)} \colon \left( x, y \right) = (-3a, -a)$,

\item
$Q^{(2)} \colon \left( x, y \right) = (-a, a)$,

\item
$Q^{(3)} \colon \left( x, y \right) = (a, 3a)$.
\end{itemize}
\begin{figure}
\begin{picture}(300, 300)

	{\thicklines
	\put(20, 45){\line(1, 0){225}}
	\put(220, 20){\line(0, 1){225}}
	\put(45, 20){\line(0, 1){225}}
	\put(20, 220){\line(1, 0){225}}
	}

	\multiput(20, 170)(5, 0){45}{\line(1, 0){2.5}}
	\multiput(20, 145)(5, 0){45}{\line(1, 0){2.5}}
	\multiput(20, 120)(5, 0){45}{\line(1, 0){2.5}}
	\multiput(20, 95)(5, 0){45}{\line(1, 0){2.5}}
	
	\multiput(170, 20)(0, 5){45}{\line(0, 1){2.5}}
	\multiput(145, 20)(0, 5){45}{\line(0, 1){2.5}}
	\multiput(120, 20)(0, 5){45}{\line(0, 1){2.5}}
	\multiput(95, 20)(0, 5){45}{\line(0, 1){2.5}}

	\put(255, 218){$y = \infty$}
	\put(255, 168){$y = 3a$}
	\put(255, 143){$y = a$}
	\put(255, 118){$y = -a$}
	\put(255, 93){$y = -3a$}
	\put(255, 43){$y = 0$}

	\put(10, 255){$x =$}
	\put(42, 255){$0$}
	\put(82, 255){$-3a$}
	\put(112, 255){$-a$}
	\put(142, 255){$a$}
	\put(164, 255){$3a$}
	\put(215, 255){$\infty$}

	\put(170, 220){\circle*{5}}
	\put(172, 225){$P^{(1)}$}
	
	\put(220, 145){\circle*{5}}
	\put(222, 150){$P^{(2)}$}
	
	\put(145, 220){\circle*{5}}
	\put(147, 225){$P^{(3)}$}
	
	\put(220, 120){\circle*{5}}
	\put(222, 125){$P^{(4)}$}
	
	\put(120, 220){\circle*{5}}
	\put(122, 225){$P^{(5)}$}
	
	\put(220, 95){\circle*{5}}
	\put(222, 100){$P^{(6)}$}
	
	\put(95, 120){\circle*{5}}
	\put(96, 108){$Q^{(1)}$}
	
	\put(120, 145){\circle*{5}}
	\put(121, 133){$Q^{(2)}$}
	
	\put(145, 170){\circle*{5}}
	\put(146, 158){$Q^{(3)}$}

\end{picture}
\caption{
Diagram showing the centers of the blow-ups needed to obtain a space of initial conditions for the mapping (\ref{eq:ex1p1p1}).
}\label{fig:ex1blowup}
\end{figure}
Then, $\varphi$ becomes an automorphism on $X$.

Let $D^{(1)}, D^{(2)} \subset X$ be the strict transforms of the lines $\{ x = \infty \}$ and $\{ y = \infty \}$ in $\mathbb{P}^1 \times \mathbb{P}^1$, respectively, and let $C^{(i)}, \widetilde{C}^{(i)} \subset X$ be the exceptional curves of the blow-up at $P^{(i)}, Q^{(i)}$, respectively.
Let $\{ y = \pm 3a \}, \{ x = \pm 3a \} \subset X$ be the strict transforms of the corresponding lines in $\mathbb{P}^1 \times \mathbb{P}^1$.
These curves move under $\varphi$ as follows:
\begin{align}
	& D^{(1)} \to D^{(2)} \to D^{(1)}, \nonumber \\
	& \{ y = 3a \} \to C^{(1)} \to C^{(2)} \to \cdots \to C^{(6)} \to \{ x = -3a \}, \label{eq:ex1sc1} \\
	& \{ y = -3a \} \to \widetilde{C}^{(1)} \to \widetilde{C}^{(2)} \to \widetilde{C}^{(3)} \to \{ x = 3a \}. \label{eq:ex1sc2}
\end{align}

We thus find an exact correspondence between the singularity pattern (\ref{eq:ex1scpattern1}) and the motion of curves (\ref{eq:ex1sc1}) on the one hand, and between the pattern (\ref{eq:ex1scpattern1}) and the motion (\ref{eq:ex1sc2}) on the other hand, as in $\mathbb{P}^1 \times \mathbb{P}^1$ these curves correspond to the points $P^{(1)}, \ldots, P^{(6)}, Q^{(1)}, Q^{(2)}, Q^{(3)}$.
After sufficient steps, however, these points again become curves.
This phenomenon corresponds to the recovery of the information on the initial value, and thus gives a geometric interpretation of singularity confinement.

We will see in \textsection\ref{sec:degree} how to calculate the degree growth of the equation from the linear action on $\operatorname{Pic} X$.
According to Takenawa \cite{takenawa1}, the maximum eigenvalue of the linear action gives the dynamical degree of the equation.
Using
\begin{align*}
	D^{(1)} + C^{(2)} + C^{(4)} + C^{(6)}
	&\sim \{ x = -3a \} + \widetilde{C}^{(1)} \\
	&\sim \{ x = 3a \} + C^{(1)},
\end{align*}
we have
\begin{align*}
	\{ x = -3a \} \sim D^{(1)} + C^{(2)} + C^{(4)} + C^{(6)} - \widetilde{C}^{(1)}, \\
	\{ x = 3a \} \sim D^{(1)} - C^{(1)} + C^{(2)} + C^{(4)} + C^{(6)},
\end{align*}
where ``$\sim$'' means the linear equivalence.
Thus, the matrix of $\varphi_{*} \colon \operatorname{Pic} X \to \operatorname{Pic} X$ with respect to the basis $D^{(1)}, D^{(2)}, C^{(1)}, \ldots, C^{(6)}, \widetilde{C}^{(1)}, \widetilde{C}^{(2)}, \widetilde{C}^{(3)}$ is
\begin{equation}\label{eq:exmatrix}
\left( \begin{matrix}
	0 & 1 &   &   &   &   &   & 1 &   &   & 1 \\
	1 & 0 &   &   &   &   &   & 0 &   &   & 0 \\
	  &   & 0 &   &   &   &   & 0 &   &   & -1 \\
	  &   & 1 & 0 &   &   &   & 1 &   &   & 1 \\
	  &   & 0 & 1 & 0 &   &   & 0 &   &   & 0 \\
	  &   &   & 0 & 1 & 0 &   & 1 &   &   & 1 \\
	  &   &   &   & 0 & 1 & 0 & 0 &   &   & 0 \\
	  &   &   &   &   & 0 & 1 & 1 &   &   & 1 \\
	  &   &   &   &   &   & 0 & -1 & 0 &  & 0 \\
	  &   &   &   &   &   &   & 0 & 1 & 0 & 0 \\
	  &   &   &   &   &   &   & 0 & 0 & 1 & 0
\end{matrix} \right).
\end{equation}
In fact, since the eigenvalues of this matrix all have modulus $1$, the entropy of this equation is $0$.
\end{example}

In the above example, we started with $\mathbb{P}^1 \times \mathbb{P}^1$ and only used blow-ups to obtain a space of initial conditions.
However, it is also possible to start with $\mathbb{P}^2$ (or a Hirzebruch surface $\mathbb{F}_a$) and, in general, blow-downs are also necessary to obtain a space of initial conditions.
If we admit the use of blow-downs, we can take an arbitrary rational surface as a starting point.
Therefore, the definition of a space of initial conditions is as follows:

\begin{definition}[space of initial conditions for autonomous systems]\label{def:sicautonomous}
If for an autonomous equation $\varphi \colon \mathbb{P}^2 \dasharrow \mathbb{P}^2$, there exist a rational surface $X$ and a birational map $f \colon X \dasharrow \mathbb{P}^2$ such that $f^{-1} \circ \varphi \circ f$ is an automorphism on $X$:
\[
\xymatrix{
	X \ar[r]^{\sim} \ar@{.>}[d]^{f} & X \ar@{.>}[d]^{f} \\
	\mathbb{P}^2 \ar@{.>}[r]^{\varphi} & \mathbb{P}^2,
}
\]
then $X$ is called a {\em space of initial conditions} for $\varphi$.
That is, an autonomous equation has a space of initial conditions if it can be regularized as an automorphism on some rational surface.
\end{definition}

It is important to note that in general $f$ is a composition of a finite number of blow-ups and blow-downs.

\begin{remark}\label{rem:bdownautonomous}
Consider an autonomous equation $\varphi \colon \mathbb{P}^2 \dasharrow \mathbb{P}^2$ with a space of initial conditions $f \colon X \dasharrow \mathbb{P}^2$ and assume that the degree of $\varphi^n$ is unbounded.
In this case, $X$ has infinitely many exceptional curves of the first kind and thus Theorem~\ref{app:nagata} implies that there exists a birational morphism $g \colon X \to \mathbb{P}^2$.
Let $\psi = g \circ f^{-1} \circ \varphi \circ f \circ g^{-1}$.
Then $\psi$ is a birational automorphism on $\mathbb{P}^2$:
\[
\xymatrix{
	\mathbb{P}^2 \ar@{.>}[rrr]^{\psi} & & & \mathbb{P}^2 \\
	& X \ar[r]^{\sim} \ar@{.>}[ld]_{f} \ar[ul]^{g} & X \ar@{.>}[rd]^{f} \ar[ur]_{g} & \\
	\mathbb{P}^2 \ar@{.>}[rrr]^{\varphi} & & & \mathbb{P}^2.
}
\]
If we identify two equations that are transformed to each other by a coordinate change of $\mathbb{P}^2$, then $\varphi$ and $\psi$ are the same equation.
Therefore, by changing coordinates on $\mathbb{P}^2$ appropriately, we can think of $f$ in Definition~\ref{def:sicautonomous} as a composition of blow-ups as long as the degree growth of the equation is unbounded.
\end{remark}

All automorphisms on rational surfaces have been classified by Gizatullin \cite{rationalgsurfaces}.
On the other hand, all birational automorphisms on surfaces have been classified by Diller and Favre \cite{df}.
Extracting the classification of birational automorphisms on rational surfaces from their theorem and interpreting it from the viewpoint of integrable systems, we have the following classification of autonomous equations of the plane:

\begin{theorem}[Diller-Favre \cite{df}]\label{thm:df}
Autonomous equations $\varphi$ of the plane are classified into the following $5$ classes:
\begin{itemize}
\item[class $1$:]
The degree of $\varphi^n$ is bounded.

This type of equation has a space of initial conditions.

For example, projective transformations on $\mathbb{P}^2$ and periodic mappings belong to this class.

\item[class $2$:]
The degree of $\varphi^n$ grows linearly.

This type of equation does not have a space of initial conditions.

Most linearizable mappings belong this class.

\item[class $3$:]
The degree of $\varphi^n$ grows quadratically.

This type of equation has a space of initial conditions.
It is an elliptic surface and $\varphi$ preserves the elliptic fibration on the surface.

For example, the QRT mappings belong to this class \cite{qrt, qrttsuda, qrtbook}.

\item[class $4$:]
The degree of $\varphi^n$ grows exponentially but the equation has a space of initial conditions.

Its Picard number is greater than $10$.

For example, the Hietarinta-Viallet equation belongs to this class.

\item[class $5$:]
The degree of $\varphi^n$ grows exponentially and the equation does not have a space of initial conditions.

``Most'' equations belong to this class.
\end{itemize}
\end{theorem}

Moreover, Diller and Favre showed that the value of the dynamical degree of an equation is quite restricted.

\begin{definition}
A reciprocal quadratic integer is a root of $\lambda^2 - a \lambda + 1 = 0$ for some integer $a$.
A real algebraic integer $\lambda > 1$ is a Pisot number if all its conjugates have modulus less than $1$.
A real algebraic integer $\lambda > 1$ is a Salem number if $1/\lambda$ is a conjugate and all (but at least one) of the other conjugates lie on the unit circle.
\end{definition}

\begin{remark}
It goes without saying that reciprocal quadratic integers greater than $1$ and Salem numbers are by definition irrational.
\end{remark}

\begin{theorem}[Diller-Favre \cite{df}]\label{thm:autonomousgrowth}
The dynamical degree of an autonomous equation of the plane is $1$, a Pisot number or a Salem number.
\end{theorem}

\begin{theorem}[Diller-Favre \cite{df}]\label{thm:sicgrowth}
If an autonomous equation of the plane has a space of initial conditions, then its dynamical degree must be $1$, a reciprocal quadratic integer greater than $1$ or a Salem number.
If the dynamical degree is $1$, then the degree growth is bounded or quadratic.
In particular, this implies that if the degree grows linearly, then the equation does not have a space of initial conditions.
\end{theorem}

Theorem~\ref{thm:sicgrowth} says that if a mapping has a space of initial conditions, then the value of its dynamical degree (and algebraic entropy) is strongly restricted.
Thus, it is sometimes possible to prove the nonexistence of a space of initial conditions by calculating the algebraic entropy \cite{extendedhv}.

It is well-known that there is a close relation between the degree growth of an equation and the Picard number of its space of initial conditions:

\begin{proposition}\label{prop:picardgrowth}
If an equation has a space of initial conditions with the Picard number less than $10$ (resp.~$11$), then its degree growth is unbounded (resp.~at most quadratic).
Moreover, if the degree growth is quadratic and a space of initial conditions is minimal (Definition~\ref{def:minimalautonomous}), then its Picard number is $10$.
\end{proposition}

All autonomous mappings with quadratic degree growth have been classified in \cite{classification2d}.

Moreover, there is a strong result about equations with bounded degree:

\begin{theorem}[Blanc-D\'{e}serti \cite{boundedp2}]
Let $\varphi$ be a nonperiodic equation with bounded degree growth and let $X$ be a space of initial conditions.
Then, $\varphi$ can be minimized from $X$ to either $\mathbb{P}^2$ or a Hirzebruch surface $\mathbb{F}_a$ with $a \ne 1$.
Furthermore, $\varphi$ is birationally conjugate to a projective transformation on $\mathbb{P}^2$.
\end{theorem}

Therefore, besides periodic mappings, all autonomous integrable (zero algebraic entropy) equations of the plane are characterized by a minimal space of initial conditions with Picard number less than $11$.

While the most famous class of nonautonomous equations that have a space of initial conditions is that of the discrete Painlev\'{e} equations, there are many other examples.
For instance, using algebro-geometric methods, Takenawa considered a nonautonomous extension of the Hietarinta-Viallet equation \cite{takenawa1, takenawa2, takenawa3}.
In addition, one of the most important and powerful methods to find a nonautonomous equation with all singularities confined is so-called late confinement, which was first reported in \cite{lateconfinement}.
This method provides us with a family of nonautonomous equations that pass the singularity confinement test \cite{deautonomisation, redeeming}.

Until now, no general theory of nonautonomous mappings with a space of initial conditions has been formulated.
One of the main aims of this paper is such a classification of integrable equations with a space of initial conditions.
It is known that all discrete Painlev\'{e} equations have a space of initial conditions (by definition \cite{sakai}) and that they are integrable (as shown by Takenawa \cite{takenawa2}).
Then, is it conceivable that there exists an integrable equation that is not a discrete Painlev\'{e} equation but which has a space of initial conditions?

The reason why there has been almost no general theory of spaces of initial conditions in the nonautonomous case is the difficulty in setting up a suitable starting point.
In the autonomous case, an equation with a space of initial conditions is reduced to one automorphism on a single rational surface.
However, even if a nonautonomous system such as a discrete Painlev\'{e} equation has a space of initial conditions, it is in general not reducible to an automorphism on a surface.
Furthermore, in the nonautonomous case, even the centers of the blow-ups and therefore the obtained surface do depend on $n$.
As a result, a space of initial conditions is not a single surface in a strict sense, but rather a family of surfaces.
Therefore, choosing appropriate $\varphi_n$, one can obtain many pathological examples.
It is true that this kind of problem does not matter when we consider a concrete example such as a discrete Painlev\'{e} equation or a nonautonomous extension of the Hietarinta-Viallet equation.
However, if we are interested in a classification, we cannot avoid the need to set up an appropriate starting point.
In \textsection\ref{sec:na}, we shall first describe several artificial examples and then define a space of initial conditions for nonautonomous equations.
We will also recall the space of initial conditions in Sakai's sense and show that these two definitions are equivalent.

\textsection\ref{sec:degree} mainly contains preliminaries.
We shall see that, under our definition of a space of initial conditions, many analogues of the properties of autonomous equations still hold.

As in the autonomous case, in order to use the Picard number of a space of initial conditions in a classification, we must consider minimizations since the Picard number of a space of initial conditions can be artificially increased.
Minimization was considered by Carstea and Takenawa \cite{ct}, but general nonautonomous cases have not been considered.
In \textsection\ref{sec:minimalization}, we shall see that a minimization of a space of initial conditions in the nonautonomous case is in fact quite similar to that in the autonomous case.

\textsection\ref{subsec:main} is the main part of this paper.
We consider a minimization of an integrable equation with unbounded degree growth and a space of initial conditions to classify all such equations.
As a result, we will obtain the main theorem of this paper (Theorem~\ref{thm:main}), which states that an integrable mapping of the plane with unbounded degree growth which possesses a space of initial conditions must be one of the discrete Painlev\'{e} equations.
We also show the uniqueness of the minimization (Proposition~\ref{prop:unique}).

\textsection\ref{subsec:exp} contains some additional results on the minimization of a space of initial conditions in the nonintegrable case.
We will not classify such equations, but instead give a procedure to minimize a general space of initial conditions and show the uniqueness of the minimization.

For the convenience of the reader, in Appendix~\ref{sec:surf}, we describe the notations we use throughout the paper and recall some basic results on algebraic surfaces.
Appendix~\ref{sec:proofoflemma} is an elementary but rather involved proof of a fundamental fact in linear algebra (Lemma~\ref{lem:symmetricjordan}).


\section{Space of initial conditions for nonautonomous systems}\label{sec:na}

In this section, we define a space of initial conditions in the nonautonomous case.
First we show several examples in order to explain what is necessary in the definition of a space of initial conditions.
Next, we will state the definition (Definition~\ref{def:sic}).
Finally, we will see in Proposition~\ref{prop:corresp} that there is a correspondence between our definition of a space of initial conditions and that of Sakai.

\begin{example}\label{ex:bupnormal}
Let us consider a nonautonomous extension of the equation in Example~\ref{ex:firstexample}:
\[
	y_{n+1} = \frac{(y_n + 3a_n - 2 \alpha)y_{n-1} - 2 a_n y_n - 14 \alpha a_n}{y_n - 3 a_n - 2 \alpha},
\]
where $\alpha$ is a constant and $a_n$ satisfies $a_{n+1} = a_n + \alpha \ne 0$ \cite{ramanieq}.
We can recover the autonomous case by taking $\alpha = 0$.
As in the autonomous case, this equation has two confined singularity patterns:
\[
	\{
	3 a_n + 2 \alpha, \infty, a_n, \infty, -a_n - 6 \alpha, \infty, -3 a_n - 16 \alpha
	\}
\]
and
\[
	\{
	-3 a_n + 2 \alpha, -a_n + 3 \alpha, a_n + 6 \alpha, 3 a_n + 11 \alpha
	\}.
\]

Let us regularize this equation as a family of isomorphisms on surfaces by blow-ups.
The equation can be written as follows:
\[
	\varphi_n \colon \mathbb{P}^1 \times \mathbb{P}^1 \dasharrow \mathbb{P}^1 \times \mathbb{P}^1,
	\quad
	(x_n, y_n) \mapsto
	\left(
	y_n,
	\frac{(y_n + 3a_n - 2 \alpha)x_n - 2a_n y_n - 14 \alpha a_n}{y_n - 3 a_n - 2 \alpha}
	\right).
\]
Using the variables $s_n = 1/x_n$ and $t_n = 1/y_n$, we cover $\mathbb{P}^1 \times \mathbb{P}^1$ with $4$ copies of $\mathbb{C}^2$:
\[
	\mathbb{P}^1 \times \mathbb{P}^1 = (x_n, y_n) \cup \left( s_n, y_n \right) \cup \left( x_n, t_n \right) \cup \left( s_n, t_n \right).
\]
Let $X_n$ be the surface obtained by blowing up $\mathbb{P}^1 \times \mathbb{P}^1$ at the following $9$ points:
\begin{itemize}
\item
$P^{(1)}_n \colon \left( x_n, t_n \right) = (3a_n - \alpha, 0)$,

\item
$P^{(2)}_n \colon \left( s_n, y_n \right) = (0, a_n - 2 \alpha)$,

\item
$P^{(3)}_n \colon \left( x_n, t_n \right) = (a_n - 3 \alpha, 0)$,

\item
$P^{(4)}_n \colon \left( s_n, y_n \right) = (0, -a_n - 2 \alpha)$,

\item
$P^{(5)}_n \colon \left( x_n, t_n \right) = (-a_n - \alpha, 0)$,

\item
$P^{(6)}_n \colon \left( s_n, y_n \right) = (0, -3a_n + 2 \alpha)$,

\item
$Q^{(1)}_n \colon \left( x_n, y_n \right) = (-3a_n + 5 \alpha, -a_n + 4 \alpha)$,

\item
$Q^{(2)}_n \colon \left( x_n, y_n \right) = (-a_n + 5 \alpha, a_n + 4 \alpha)$,

\item
$Q^{(3)}_n \colon \left( x_n, y_n \right) = (a_n + 3 \alpha, 3a_n + 2 \alpha)$.
\end{itemize}
The configuration of these 9 points is almost the same as in the autonomous case (Figure~\ref{fig:ex1blowup}).
Since
\[
	\varphi_n \left( P^{(i)}_n \right) = P^{(i+1)}_{n+1}
\]
for $i = 1, \ldots, 5$ and
\[
	\varphi_n \left( Q^{(i)}_n \right) = Q^{(i+1)}_{n+1}
\]
for $i = 1, 2$, one finds that $\varphi_n$ is indeed an isomorphism from $X_n$ to $X_{n+1}$.

As in the autonomous case, let us label specific curves on $X_n$ as follows:
\begin{itemize}
\item
$D^{(1)}_n, D^{(2)}_n$: the strict transforms of the lines $\{ x = \infty \}$ and $\{ y = \infty \}$ in $\mathbb{P}^1 \times \mathbb{P}^1$, respectively.

\item
$C^{(i)}_n$: the exceptional curve of the blow-up at $P^{(i)}_n$.

\item
$\widetilde{C}^{(i)}_n$: the exceptional curve of the blow-up at $Q^{(i)}_n$.

\end{itemize}
These curves move under the equation as:
\begin{align*}
	& D^{(1)} \to D^{(2)} \to D^{(1)}, \\
	& \{ y = 3a + 2 \alpha \} \to C^{(1)} \to C^{(2)} \to \cdots \to C^{(6)} \to \{ x = -3a + 5 \alpha \}, \\
	& \{ y = -3a + 2 \alpha \} \to \widetilde{C}^{(1)} \to \widetilde{C}^{(2)} \to \widetilde{C}^{(3)} \to \{ x = 3a - \alpha \},
\end{align*}
where we omit the index $n$.
Thus, the matrix of $\varphi_{*} \colon \operatorname{Pic} X_n \to \operatorname{Pic} X_{n+1}$ with respect to the basis $(D^{(1)}, D^{(2)}, C^{(1)}, \ldots, C^{(6)}, \widetilde{C}^{(1)}, \widetilde{C}^{(2)}, \widetilde{C}^{(3)})$ coincides exactly with the one obtained in the autonomous case (\ref{eq:exmatrix}).
Therefore, the algebraic entropy of this equation is zero as well.
\end{example}

As we have seen in the above example, when considering a space of initial conditions, it is most important for the equation to be regularized as a (family of) isomorphism(s) on surfaces.
However, since there exist lots of pathological nonautonomous equations, this condition is so weak that we cannot hope to say anything about general properties of such equations.

Let us consider some of these examples.
In the following examples, we fix the starting index at $n = 0$, i.e.\,by $\deg \psi^n$ we denote $\deg (\psi_{n-1} \circ \cdots \circ \psi_0)$ (Remark~\ref{rem:degreeexist}).

\begin{example}\label{ex:power}
Let $\varphi$ be an arbitrary autonomous equation with unbounded degree growth and a space of initial condition $X$ (for example, the mapping $\varphi$ in Example~\ref{ex:firstexample}), and let $(d_n)_{n>0}$ be an arbitrary sequence of positive integers.
Define sequences $(p_n)_{n \ge 0}$ and $(q_n)_{n>0}$ by
\[
	p_0 = 0, \quad
	p_n = \max \{ k \in \mathbb{Z}_{\ge 0} \, | \, \deg \varphi^{k} \le d_n \}, \quad
	q_n = p_n - p_{n-1}.
\]
Let
\[
	\psi_n = \varphi^{q_n} \colon \mathbb{P}^2 \dasharrow \mathbb{P}^2
\]
for all $n > 0$.
Then, we have
\[
	\deg (\psi_n \circ \cdots \circ \psi_1) = \deg \varphi^{p_n} \approx d_n.
\]
Since $\varphi$ is an automorphism on $X$, so is $\psi_n$ for all $n$.
Therefore, by choosing $(d_n)_n$ appropriately, we can construct many equations that can be reduced to families of isomorphisms (automorphisms) on surfaces but that have arbitrary degree growth.

\noindent
\underline{Case 1}

Let $\lambda$ be an arbitrary real number greater than $1$ and let $d_n$ be the greatest integer not greater than $\lambda^n$.
In this case, the entropy of the mapping $\psi_n$ is $\log \lambda$.

\noindent
\underline{Case 2}

Let $\lambda$ as in Case 1 and let
\[
	d_n = \begin{cases}
		\text{the greatest integer not greater than } \lambda^n & (n \colon \text{even}) \\
		1 & (n \colon \text{odd}).
	\end{cases}
\]
In this case, the entropy of the mapping $\psi_n$ does not exist.
If we change the definition of the entropy to
\[
	\limsup_{n \to \infty} \frac{1}{n} \log \left( \deg \psi^n \right),
\]
then the entropy exists and is $\log \lambda$.

\noindent
\underline{Case 3}

Let $d_n = n$.
In this case, the degree of $\psi^n$ grows linearly but the equation can be reduced to a family of automorphisms on $X$ (in contrast to Theorem~\ref{thm:df}, class~2 for autonomous mappings).

\noindent
\underline{Case 4}

Let $d_n$ grow faster than any exponential function of $n$, for example $d_n = n^n$.
In this case, the entropy of the mapping $\psi_n$ is $+\infty$ (in contrast to Remark~\ref{rem:degreeexist} for autonomous mappings).
\end{example}

\begin{example}\label{ex:nonpower}
The same technique as above can also be used in case the original $\varphi$ does not have a space of initial conditions.
Let $\varphi$ be an autonomous equation with unbounded degree growth but no space of initial conditions (for example a linearizable mapping) and let $d_n = n^2$.
Then, we obtain a mapping $\psi_n$ that has a quadratic degree growth but cannot be regularized as a family of isomorphisms on surfaces.
\end{example}

\begin{example}
In the above two examples, the equations are quite artificial and practically impossible to write explicitly.
Usually, the term ``nonautonomous equation'' refers to an equation with several nonautonomous coefficients such as Example~\ref{ex:bupnormal}.
However, even in this class of equations, there are strange mappings.

Consider the equation
\[
	x_{n+1} = a_n x^2_n + (1 - a_n)x_n + b x_{n-1},
\]
where $b$ is a general constant and $a_n$ is a nonautonomous coefficient.
We are interested only in the case $a_n = 0, 1$.

In the case where $a_n$ is always $0$, this equation is a linear mapping and thus the degree growth is obviously bounded.
On the other hand, in the case where $a_n$ is always $1$, this equation is a H\'{e}non map \cite{henon} and its algebraic entropy is $\log 2$.

If $a_n$ can take both values $0$ and $1$, then these two cases are mixed.
It is obvious that for any real number $\lambda \in [1, 2]$, there exist a sequence $(a_n)_n$ such that the dynamical degree of the above equation is $\lambda$.

It is always possible to mix two different equations by using one nonautonomous coefficient.
For example, if we start with two autonomous equations that have the same space of initial conditions, then the mixed equation is reduced to a family of automorphisms on a surface but exhibits strange behavior.
\end{example}

What is important is that, even if the obtained surfaces and isomorphisms depend on $n$, their ``fundamental structures'' (for example, the intersection pattern of specific curves and the linear action induced on the Picard groups) are the same.
When we consider a concrete equation such as Example~\ref{ex:bupnormal}, it is (in principle) possible to check whether those structures do or do not depend on $n$.
However, it is difficult to define mathematically what constitutes a fundamental structure for general equations.
In this paper we shall therefore define a space of initial conditions as follows:

\begin{definition}[space of initial conditions for nonautonomous systems]\label{def:sic}
An equation $\varphi_n \colon \mathbb{P}^2 \dashrightarrow \mathbb{P}^2$ has a {\em space of initial conditions} if (after an appropriate coordinate change) the following three conditions are satisfied:
\begin{itemize}
\item
There exists a composition of blow-ups $\pi_n = \pi^{(1)}_n \circ \cdots \circ \pi^{(r)}_n \colon X_n \rightarrow \mathbb{P}^2$ for each $n$ such that the induced birational maps $\varphi_n \colon X_n \dashrightarrow X_{n+1}$ are all isomorphisms:
\[
\xymatrix{
	\ar[r] & X_{n-1} \ar[d]^{\pi_{n-1}} \ar[r]^{\sim} & X_n \ar[d]^{\pi_n} \ar[r]^{\sim} & X_{n+1} \ar[d]^{\pi_{n+1}} \ar[r] & \\
	\ar@{.>}[r] & \mathbb{P}^2 \ar@{.>}[r]^{\varphi_{n-1}} & \mathbb{P}^2 \ar@{.>}[r]^{\varphi_n} & \mathbb{P}^2 \ar@{.>}[r] & .
}
\]
\item
Let $e_n = (e^{(0)}_n, \ldots, e^{(r)}_n)$ be the geometric basis corresponding to $\pi_n$ (Definition~\ref{app:geometricbasis}).
Then, the matrices of $\varphi_{n*} \colon \operatorname{Pic} X_n \to \operatorname{Pic} X_{n+1}$ with respect to these bases do not depend on $n$.
\item
The set of all effective classes in $\operatorname{Pic} X_n$ does not depend on $n$, i.e.\,if $\sum_i a^{(i)} e^{(i)}_n \in \operatorname{Pic} X_n$ is effective, then so is $\sum_i a^{(i)} e^{(i)}_{k} \in \operatorname{Pic} X_{k}$ for any $k$.
\end{itemize}
\end{definition}

We refer the reader to Appendix~\ref{sec:surf} for an explanation of the notations and for a summary of some basic results on algebraic surfaces.

Note that in the nonautonomous case, a space of initial conditions does not consist of a single surface but of a family of surfaces.
It also contains information about the centers and ordering of the blow-ups.

\begin{remark}
As in the autonomous case (Remark~\ref{rem:bdownautonomous}), blow-downs are necessary, in general, to construct a space of initial conditions.
However, to avoid unnecessary complexity, we use the phrase ``after an appropriate coordinate change'' instead.
We will see in Remark~\ref{rem:bdownna} the rigorous definition including blow-downs.
\end{remark}

Usual nonconfining equations such as linearizable mappings and H\'{e}non maps do not satisfy the first condition in Definition~\ref{def:sic}.
On the other hand, Example~\ref{ex:power} does satisfy the first and third conditions but does not satisfy the second.

The third condition imposes some constraint on the centers and ordering of blow-ups.
Unfortunately, it is not easy in general to check the third condition in Definition~\ref{def:sic} for a concrete equation.
However, we shall see that even if only the first and second conditions are satisfied, we can still calculate the degree growth by Proposition~\ref{prop:jordan} (since its proof does not need the third condition).
One reason why we introduce the third condition is the correspondence to a space of initial conditions in Sakai's sense, which we shall introduce later.

\begin{remark}\label{rem:phi}
Let us first have a closer look at the second condition.
Let
\[
	\mathbb{Z}^{1, r} = \mathbb{Z}e^{(0)} \oplus \cdots \oplus \mathbb{Z}e^{(r)}
\]
and define on $\mathbb{Z}^{1, r}$ a symmetric bilinear form $(- , -)$ by
\[
	(e^{(i)}, e^{(j)}) = 
	\begin{cases}
		1 & (i = j = 0) \\
		-1 & (i = j \ne 0) \\
		0 & (i \ne j).
	\end{cases}
\]
Let
\[
	\iota_n \colon \mathbb{Z}^{1, r} \to \operatorname{Pic} X_n, \quad
	e^{(i)} \mapsto e^{(i)}_n
\]
and $\Phi_n = \iota^{-1}_{n+1} \varphi_{n*} \iota_n$:
\[
\xymatrix{
	\ar[r] & \mathbb{Z}^{1, r} \ar[r]^{\Phi_{n-1}} \ar[d]^{\iota_{n-1}} & \mathbb{Z}^{1, r}\ar[r]^{\Phi_n} \ar[d]^{\iota_n} & \mathbb{Z}^{1, r}\ar[r] \ar[d]^{\iota_{n+1}} & \\
	\ar[r] & \operatorname{Pic X_{n-1}} \ar[r]^{\varphi_{n-1*}} & \operatorname{Pic X_n} \ar[r]^{\varphi_{n*}} & \operatorname{Pic X_{n+1}} \ar[r] & .
}
\]
Then, the meaning of the second condition is that $\Phi_n$ does not depend on $n$.
We then simply denote $\Phi_n$ by $\Phi$.

We will use these notations in \textsection\ref{sec:degree}.
\end{remark}

\begin{lemma}
Let $K = \iota^{-1}_n K_{X_n} = - 3e^{(0)} + e^{(1)} + \cdots + e^{(r)}$.
Then $\Phi$ preserves $K$ and $(-, -)$, i.e.
\[
	\Phi K = K, \quad
	(v, w) = (\Phi v, \Phi w)
\]
for all $v, w \in \mathbb{Z}^{1, r}$.
\end{lemma}

\begin{proof}
Immediate from the fact that $\varphi_{n*}$ preserves the canonical class and the intersection number on the surface.
\end{proof}

Next, we review the notion of a space of initial conditions in Sakai's sense.

Let $X$ be a basic rational surface (Definition~\ref{app:rational}).
Let $e = (e^{(0)}, \ldots, e^{(r)}), \widetilde{e} = (\widetilde{e}^{(0)}, \ldots, \widetilde{e}^{(r)})$ be geometric bases and $\pi, \widetilde{\pi} \colon X \to \mathbb{P}^2$ the corresponding birational morphisms.
Then we obtain a birational automorphism $\widetilde{\pi} \circ \pi^{-1} \colon \mathbb{P}^2 \dashrightarrow \mathbb{P}^2$ which will become a part of an equation.

Let $\sigma$ be the $\mathbb{Z}$-linear map on $\operatorname{Pic} X$ defined by
\[
	e^{(0)} \mapsto \widetilde{e}^{(0)}, \ldots, e^{(r)} \mapsto \widetilde{e}^{(r)}.
\]
Suppose that $\sigma^n e = (\sigma^n e^{(0)}, \ldots, \sigma^n e^{(r)})$ is a geometric basis for each $n$ and let $\pi_n \colon X \to \mathbb{P}^2$ be the corresponding birational morphism.
Then, we obtain the equation
\[
	\varphi_n = \pi_{n+1} \circ \pi^{-1}_n \colon \mathbb{P}^2 \dasharrow \mathbb{P}^2.
\]

\begin{example}\label{ex:simplestcremona}
The following is probably the simplest example where $\sigma^n e$ is not a geometric basis.

We cover $\mathbb{P}^2$ by three copies of $\mathbb{C}^2$ as follows:
\[
	\mathbb{P}^2 = \left( x, y \right) \cup \left( \frac{x}{y}, \frac{1}{y} \right) \cup \left( \frac{y}{x}, \frac{1}{x} \right).
\]
Let $\pi^{(1)}, \pi^{(2)}, \pi^{(3)}$ be the blow-ups at the following points:
\begin{itemize}
\item
$\pi^{(1)}$: at $P^{(1)} \colon (x, y) = (0, 0)$,
\item
$\pi^{(2)}$: at $P^{(2)} \colon \left( \dfrac{1}{y}, \dfrac{x}{y} \right) = (0, 0)$,
\item
$\pi^{(3)}$: at $P^{(3)} \colon \left( \dfrac{1}{x}, \dfrac{x}{y} \right) = (0, 0)$.
\end{itemize}
Let $X$ be the surface obtained by the blow-ups $\pi = \pi^{(1)} \circ \pi^{(2)}  \circ \pi^{(3)}$ (Figure~\ref{fig:simplestcremona}) and let $e = (e^{(0)}, e^{(1)}, e^{(2)}, e^{(3)})$ be the corresponding geometric basis.
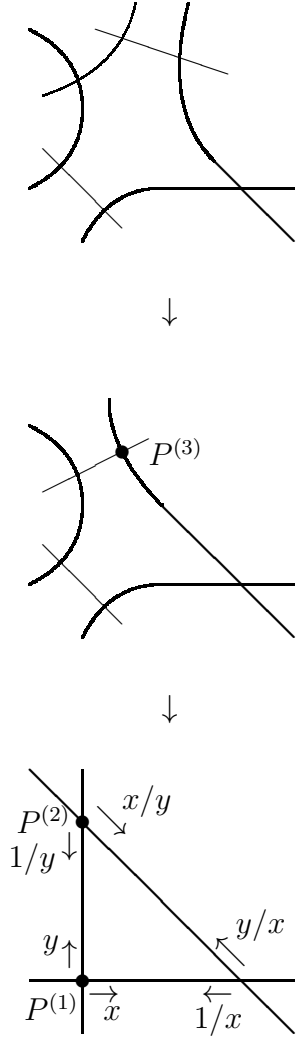
\begin{figure}
\begin{picture}(120, 410)
	{\thicklines
	\put(10, 30){\line(1, 0){100}}
	\put(30, 10){\line(0, 1){100}}
	\put(10, 110){\line(1, -1){100}}
	}
	
	\put(30, 30){\circle*{5}}
	\put(8, 15){$P^{(1)}$}
	
	\put(30, 90){\circle*{5}}
	\put(5, 86){$P^{(2)}$}

	\put(32, 22){$\rightarrow$}
	\put(38, 15){$x$}
	
	\put(22, 37){$\uparrow$}
	\put(15, 42){$y$}

	\put(75, 22){$\leftarrow$}
	\put(72, 12){$1/x$}
	
	\put(80, 38){$\nwarrow$}
	\put(88, 48){$y/x$}

	\put(22, 78){$\downarrow$}
	\put(2, 73){$1/y$}

	\put(35, 88){$\searrow$}
	\put(45, 95){$x/y$}

	\put(60, 130){$\downarrow$}

	{\thicklines
	\put(60, 180){\line(1, 0){50}}
	\put(60, 210){\line(1, -1){50}}
	\qbezier(60, 180)(40, 180)(30, 160)
	\qbezier(30, 210)(30, 190)(10, 180)
	\qbezier(30, 210)(30, 230)(10, 240)
	\qbezier(60, 210)(40, 230)(40, 250)
	}
	
	\put(15, 195){\line(1, -1){30}}
	\put(15, 215){\line(2, 1){40}}
	
	\put(45, 230){\circle*{5}}
	\put(55, 225){$P^{(3)}$}

	\put(60, 280){$\downarrow$}

	{\thicklines
	\put(60, 330){\line(1, 0){50}}
	\put(80, 340){\line(1, -1){30}}
	\qbezier(60, 330)(40, 330)(30, 310)
	\qbezier(30, 360)(30, 340)(10, 330)
	\qbezier(30, 360)(30, 380)(10, 390)
	\qbezier(80, 340)(60, 360)(70, 400)
	}
	
	\put(15, 345){\line(1, -1){30}}
	
	\qbezier(15, 365)(45, 375)(50, 400)
	\put(35, 390){\line(3, -1){50}}

\end{picture}
\caption{
Diagram showing the blow-ups needed to obtain $X$ in Example~\ref{ex:simplestcremona}.
}
\label{fig:simplestcremona}
\end{figure}

It is obvious that
\[
	\widetilde{e} =
	(\widetilde{e}^{(0)}, \widetilde{e}^{(1)}, \widetilde{e}^{(2)}, \widetilde{e}^{(3)}) =
	(e^{(0)}, e^{(2)}, e^{(3)}, e^{(1)})
\]
is another geometric basis on $\operatorname{Pic} X$.
Let $\sigma$ be the $\mathbb{Z}$-linear transformation on $\operatorname{Pic} X$ defined by $e^{(i)} \mapsto \widetilde{e}^{(i)}$ for all $i$.
While $e$ and $\sigma e = \widetilde{e}$ are geometric, $\sigma^2 e = (e^{(0)}, e^{(3)}, e^{(1)}, e^{(2)})$ is not since $e^{(2)} - e^{(3)}$ is effective.

It is obvious that all problems in this case come from the ordering of the $e^{(i)}$.
\end{example}

As in the above example, $\sigma^n e$ is not always geometric.
Therefore, it is necessary to impose some condition on $\sigma$.

\begin{definition}[Cremona isometry \cite{lo, dol2, sakai}]\label{def:cremona}
Let $X$ be a rational surface and let $\sigma$ be an invertible $\mathbb{Z}$-linear transformation on $\operatorname{Pic} X$.
$\sigma$ is said to be a {\em Cremona isometry} if it satisfies the following three conditions:
\begin{itemize}
\item
$\sigma$ preserves the intersection number on $\operatorname{Pic} X$, i.e.\,$F_1 \cdot F_2 = (\sigma F_1) \cdot (\sigma F_2)$ for all $F_1, F_2 \in \operatorname{Pic} X$,
\item
$\sigma$ preserves $K_X$,
\item
$\sigma$ preserves the set of effective classes, i.e.\,if $F$ is effective, then so is $\sigma F$ (and $\sigma^{-1} F$).

\end{itemize}
\end{definition}

\begin{example}
Let $\varphi$ be an automorphism on a rational surface.
Then the induced linear transformations $\varphi^{*}$ and $\varphi_{*}$ are Cremona isometries.
\end{example}

It is clear from the definition that the following holds.

\begin{lemma}\label{lem:cremonacone}
Cremona isometries preserve the nef cone.
\end{lemma}

It should be noted that, while an automorphism on a surface determines the motion of each curve, a Cremona isometry does {\em not}.
It only determines the motion of the classes of curves.
However, as shown in the following lemma, if an irreducible curve has a negative self-intersection, then its motion is completely determined.

\begin{lemma}\label{lem:crirred}
Let X be a rational surface and $\sigma$ a Cremona isometry, and let $C$ be an irreducible curve in $X$ with negative self-intersection.
Then there exists only one effective divisor $D$ such that $[D] = \sigma [C]$.
Moreover, $D$ is a prime divisor, i.e.\,an irreducible curve.
In particular, $\sigma$ acts as a permutation on the set of all exceptional curves of the first kind.
\end{lemma}

\begin{proof}
Let
\[
	\sigma [C] = \left[ \sum^{k}_{i = 1} m_i C_i \right],
\]
where $C_i$ are irreducible curves.
Since
\[
	[C] = \sum^{k}_{i = 1} m_i \sigma^{-1}[C_i]
\]
and $\sigma^{-1}[C_i]$ are all effective, it follows from Proposition~\ref{app:irrednega} that $k = 1$ and $m_1 = 1$.
\end{proof}

\begin{lemma}\label{lem:geombasis}
Let $\sigma$ be a Cremona isometry.
If $e = (e^{(0)}, \ldots, e^{(r)})$ is a geometric basis on $\operatorname{Pic} X$, then so is $\sigma e = (\sigma e^{(0)}, \ldots, \sigma e^{(r)})$.
\end{lemma}

\begin{proof}
Let $\pi = \pi^{(1)} \circ \cdots \circ \pi^{(r)} \colon X \to \mathbb{P}^2$ be the composition of blow-ups corresponding to $e$ and let $C_1, \ldots, C_r \subset X$ be the irreducible curves contracted by $\pi$.
Since all these curves have negative self-intersection, by Lemma~\ref{lem:crirred}, their motions are determined by $\sigma$.
Let us denote them by $C'_1, \ldots, C'_r$.
Since $C_i \cdot C_j = C'_i \cdot C'_j$ for all $i, j$, it is possible to contract $C'_1, \ldots, C'_r$ in the same order as $C_1, \ldots, C_r$.
It is clear that the geometric basis corresponding to this contraction is $\sigma e$.
\end{proof}

Let us see how to obtain an equation from a Cremona isometry \cite{sakai}.

\begin{definition}\label{def:sicsakai}
Let $X$ be a basic rational surface and let $\sigma$ be a Cremona isometry on $\operatorname{Pic} X$ and take an arbitrary geometric basis $e = (e^{(0)}, \ldots, e^{(r)})$.
By Lemma~\ref{lem:geombasis}, $\sigma^{n} e$ is a geometric basis for each $n$.
Let $\pi_n$ be the corresponding birational morphism to $\mathbb{P}^2$ and let $\varphi_n = \pi_{n+1} \circ \pi^{-1}_n$.
Thus we obtain $(\varphi_n)_{n \in \mathbb{Z}}$, a family of birational automorphisms on $\mathbb{P}^2$:
\[
\xymatrix{
	 &  & X \ar[ld]_{\pi_{n-1}} \ar[d]^{\pi_n} \ar[rd]^{\pi_{n+1}} &  & \\
	\ar@{.>}[r] & \mathbb{P}^2 \ar@{.>}[r]_{\varphi_{n-1}} & \mathbb{P}^2 \ar@{.>}[r]_{\varphi_n} & \mathbb{P}^2 \ar@{.>}[r] & .
}
\]
This is the equation defined by $X, \sigma$ and $e$, and we call $X$ a space of initial conditions (in Sakai's sense).
Since the choice of $e$ only determines the specific coordinates, we sometimes think of $(X, \sigma)$ as the equation itself.
\end{definition}

Note that $(\varphi_n)_{n \in \mathbb{Z}}$ is determined by $X, \sigma$ and $e$ up to an automorphism of $\mathbb{P}^2$ for each $n$, i.e.\,if $(\varphi'_n)_{n \in \mathbb{Z}}$ is another family of birational automorphisms on $\mathbb{P}^2$ defined by $X, \sigma$ and $e$, then there exist $f_n \in \operatorname{PGL}(3, \mathbb{C})$ such that $\varphi_n = f_{n+1} \circ \varphi'_n \circ f^{-1}_{n}$.

\begin{definition}[generalized Halphen surface \cite{sakai}]
A rational surface $X$ is called a {\em generalized Halphen surface} if it satisfies the following two conditions:
\begin{itemize}
\item
$- K_X$ is effective,

\item
All components of $- K_X$ are orthogonal to $- K_X$, i.e.\,$D_i \cdot (- K_X) = 0$ for any $\sum_{i} m_i D_i \in |- K_X|$.

\end{itemize}
\end{definition}

\begin{lemma}[Proposition 2 in \cite{sakai}]
Any generalized Halphen surface is a basic rational surface.
\end{lemma}

\begin{definition}[discrete Painlev\'{e} equation \cite{sakai}]\label{def:discretepainleve}
Let $X$ be a generalized Halphen surface and $\sigma$ a Cremona isometry on $\operatorname{Pic} X$ of infinite order.
Then, the equation obtained by the above procedure is called a {\em discrete Painlev\'{e} equation}.
\end{definition}

\begin{remark}
Note that according to this definition, autonomous mappings such as the QRT mappings are also labeled ``discrete Painlev\'{e}.''
\end{remark}

Using generalized Halphen surfaces, Sakai has classified (and, in a sense, defined) all discrete Painlev\'{e} equations.
Since we do not need such a detailed classification in this paper, we only give a brief summary.

The orthogonal lattice $K^{\perp}_X \subset \operatorname{Pic} X$, which is preserved under $\sigma$, is an affine root lattice of type $E^{(1)}_8$.
If $\dim |- K_X| = 0$, then the expression $\sum_i m_i D_i \in |- K_X|$ is unique.
Therefore, $\sigma$ acts on the set $\{ D_i \}_i$ as a permutation and preserves the lattice $\operatorname{span}_{\mathbb{Z}} D_i$ and its orthogonal compliment.
These two lattices are both affine root sublattices of $K^{\perp}_X$ and play an important role in the classification of the discrete Painlev\'{e} equations.

\begin{remark}
While Cremona isometries can be defined for any rational surface, we only consider basic rational surfaces such as in Definition~\ref{def:sicsakai}.
Although it is possible to consider a family of blow-downs from a nonbasic rational surface to a Hirzebruch surface $\operatorname{F}_a$ ($a \ge 2$) instead of $\mathbb{P}^2$, Theorem~\ref{app:nagata} implies that the degree growth of such an equation must be bounded.
Hence, it is sufficient to consider only basic rational surfaces as long as we are interested in equations with unbounded degree growth.
\end{remark}

Now let us clarify the correspondence between the two definitions of a space of initial conditions we considered.

\begin{proposition}\label{prop:corresp}
The two definitions of a space of initial conditions, Definition~\ref{def:sic} and Definition~\ref{def:sicsakai}, are equivalent.
\end{proposition}

\begin{proof}
First, consider the situation in  Definition~\ref{def:sic}.
Let $X = X_0$ and $\sigma$ be the $\mathbb{Z}$-linear transformation on $\operatorname{Pic} X$ defined by
\[
	\sigma = \iota_0 \Phi^{-1} \iota^{-1}_0.
\]
Then, a direct calculation shows that
\[
	\sigma^{\ell} e^{(i)}_0 = \varphi^{*}_0 \cdots \varphi^{*}_{\ell-1} e^{(i)}_{\ell}, \quad
	\sigma^{-\ell} e^{(i)}_0 = \varphi_{-1*} \cdots \varphi_{-\ell*} e^{(i)}_{-\ell}
\]
for all $\ell > 0$.

Let us show that $\sigma$ is a Cremona isometry.
It is clear, by construction, that $\sigma$ satisfies the first and second conditions on a Cremona isometry.
Let $F = \sum_i a^{(i)} e^{(i)}_0 \in \operatorname{Pic} X$ be an effective class.
Then we have
\begin{align*}
	\sigma F & = \sum_i a^{(i)} \varphi^{*}_{0} e^{(i)}_{1} \\
	& = \varphi^{*}_{0} \left( \sum_i a^{(i)} e^{(i)}_{1} \right).
\end{align*}
The third condition in Definition~\ref{def:sic} implies that $\sum_i a^{(i)} e^{(i)}_{1}$ is effective.
Since $\varphi^{*}_{0}$ preserves the set of effective classes, $\sigma F$ is also effective.
We can prove the effectiveness of $\sigma^{-1} F$ in the same way.

Next, consider the situation in  Definition~\ref{def:sicsakai}.
That is, $X$ is a basic rational surface and $\sigma$ is a Cremona isometry on $\operatorname{Pic} X$.
Take $e = (e^{(0)}, \ldots, e^{(r)})$ as a geometric basis on $\operatorname{Pic} X$ and consider the equation defined by $X, \sigma$ and $e$.
Let us recover the above situation from these data.

Let $X_n = X$ and $e_n = \sigma^n e$ for all $n \in \mathbb{Z}$.
While the $X_n$ themselves are all the same, the bases $e_n$ vary depending on $n$.
Then we have the following diagram:
\[
\xymatrix{
	 & X_{n-1} \ar@{=}[d] & X_n \ar@{=}[d] & X_{n+1} \ar@{=}[d] & \\
	\ar[r] & X \ar[d]^{\pi_{n-1}} \ar[r]^{\operatorname{id}} & X \ar[d]^{\pi_n} \ar[r]^{\operatorname{id}} & X \ar[d]^{\pi_{n+1}} \ar[r] & \\
	\ar@{.>}[r] & \mathbb{P}^2 \ar@{.>}[r]^{\varphi_{n-1}} & \mathbb{P}^2 \ar@{.>}[r]^{\varphi_n} & \mathbb{P}^2 \ar@{.>}[r] & .
}
\]
It is important to note that, while the morphisms from $X_n$ to $X_{n+1}$ are all the identity map on $X_n = X_{n+1} = X$, the $\varphi_n$ are not the identity map on $\mathbb{P}^2$ in general.

Let us check the second condition in Definition~\ref{def:sic}.
Let $A_n = \left( a^{(i, j)}_n \right)_{i, j}$ be the matrix representation of $\varphi_{n*}$ with respect to the bases $e_n$ and $e_{n+1}$.
Since $\varphi_{n*} = \operatorname{id}_{\operatorname{Pic} X}$, $A_n$ are determined by
\[
	e^{(i)}_n = \sum_j a^{(j, i)}_n e^{(j)}_{n+1}.
\]

Applying $\sigma^{k}$ and using $\sigma^{k} e^{(i)}_n = e^{(i)}_{n+k}$, we have
\[
	e^{(i)}_{n+k} = \sum_j a^{(j, i)}_n e^{(j)}_{n+k+1},
\]
which shows that $A_n$ do not depend on $n$.

Finally we check the third condition in Definition~\ref{def:sic}.
Let $F = \sum_i a^{(i)} e^{(i)}_n \in \operatorname{Pic} X$ be an effective class.
Since $\sigma^{k}$ preserves the effective classes,
\[
	\sigma^{k} F = \sum_i a^{(i)} e^{(i)}_{n+k}
\]
are effective for all $k$.
Hence the set of effective classes does not depend on $n$.
\end{proof}

We have seen that the two definitions of a space of initial conditions are equivalent.
In this paper, we will use both definitions depending on the situation.

\begin{remark}\label{rem:bdownna}
If we consider blow-downs instead of an ``appropriate coordinate change'' in Definition~\ref{def:sic}, we must assume that the blow-downs do not depend on $n$.
In this case, one possible rigorous definition is as follows:

An equation $(\varphi_n)_n$ has a space of initial conditions if there exist rational surfaces $Y_n$ and $X_n$, blow-ups $\pi_n = \pi^{(1)}_n \circ \cdots \circ \pi^{(r)}_n \colon Y_n \to \mathbb{P}^2$ and blow-downs $\epsilon_n = \epsilon^{(1)}_n \circ \cdots \circ \epsilon^{(r')}_n \colon Y_n \to X_n$ for each $n$, such that the following four conditions are satisfied:
\begin{itemize}
\item
$\varphi_n$ is an isomorphism from $X_n$ to $X_{n+1}$.

\item
Let $\widetilde{e}_n = (\widetilde{e}^{(0)}_n, \ldots, \widetilde{e}^{(r)}_n)$ be the geometric basis corresponding to $\pi_n$ and identify all $\operatorname{Pic} Y_n$ by these bases.
Let $E^{(k)}_n$ be the total transform of the exceptional class of $\epsilon^{(k)}_n$.
Then, $E^{(k)}_n$ does not depend $n$.

\item
Take a basis $e_n = (e^{(0)}_n, \ldots, e^{(r-r'+1)}_n)$ of $\operatorname{Pic} X_n$ for each $n$ such that $\epsilon^{*}_n e^{(i)}_n$ does not depend on $n$ (under the above identification).
Identify all $\operatorname{Pic} X_n$ by these bases.
Then, $\varphi_{n*}$ does not depend on $n$.

\item
The set of effective classes in $\operatorname{Pic} X_n$ (and in $\operatorname{Pic} Y_n$) does not depend on $n$ (under the above identification).
\end{itemize}

As in the autonomous case (Remark~\ref{rem:bdownautonomous}), if the equation has unbounded degree growth, then it is possible to reduce the above situation to that in Definition~\ref{def:sic} by taking new blow-downs $X_n \to \mathbb{P}^2$ (Figure~\ref{fig:bdownna}).
Needless to say, the new blow-downs must be such that the geometric basis on $\operatorname{Pic} X_n$ does not depend on $n$.
As in the autonomous case, the existence of such blow-downs is guaranteed by Theorem~\ref{app:nagata}.
Hence, as long as we are interested only in performing a classification, we may only consider the situation in Definition~\ref{def:sic}.

The reason why this kind of problem arises is that we start from a specific equation $(\varphi_n)_n$, whereas if we start from the situation in Definition~\ref{def:sicsakai}, this kind of problem does not appear.

From now on, we shall assume that a space of initial conditions is obtained only by blow-ups, i.e.\,we shall simply consider the situation in Definition~\ref{def:sic} or Definition~\ref{def:sicsakai}.
\begin{figure}
\[
\xymatrix{
	& Y_{n-1} \ar[dd] \ar[rd] & & Y_n \ar[dd] \ar[rd] & & Y_{n+1} \ar[dd] \ar[rd] & & \\
	\ar[rr] & & X_{n-1} \ar[dd] \ar[rr] & & X_{n} \ar[dd] \ar[rr] & & X_{n+1} \ar[dd] \ar[r] & \\
	 \ar@{.>}[r] & \mathbb{P}^2 \ar@{.>}[rr] & & \mathbb{P}^2 \ar@{.>}[rr] & & \mathbb{P}^2 \ar@{.>}[rr] & & \\
	 \ar@{.>}[rr] & & \mathbb{P}^2 \ar@{.>}[rr] & & \mathbb{P}^2 \ar@{.>}[rr] & & \mathbb{P}^2 \ar@{.>}[r] &
}
\]
\caption{
Diagram showing a space of initial conditions in the case where we consider blow-downs.
The second row from the bottom represents the original equation and the bottom row represents a new equation obtained by an appropriate coordinate change.
}
\label{fig:bdownna}
\end{figure}
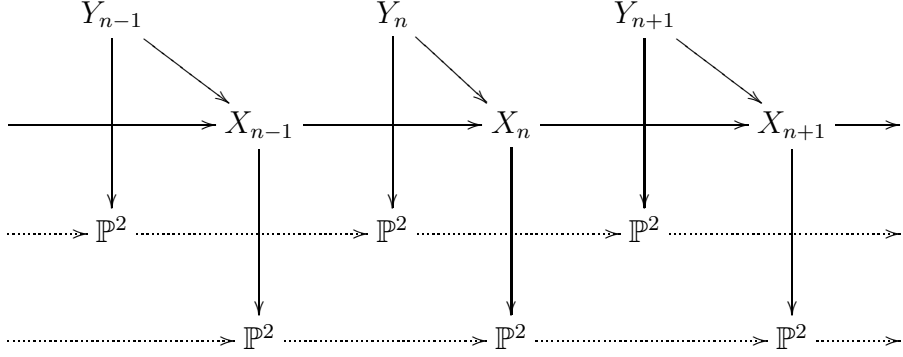
\end{remark}


\section{Basic properties of an equation with a space of initial conditions}\label{sec:degree}

In this section, we first recall Takenawa's result on the degree growth for an equation \cite{takenawa2}.
Next we shall see that, as in the autonomous case, the degree growth of a nonautonomous equation with a space of initial conditions can be classified into three cases.
Finally we show some relations between the degree growth of an equation and the Picard number of a space of initial conditions.

In this section, we consider the situation in Definition~\ref{def:sic}.
We will also use the $\Phi$ and $\iota_n$ defined in Remark~\ref{rem:phi}.

Since we will not use the third condition in Definition~\ref{def:sic} in this section, the results will still hold in the case where the third condition is not satisfied.

\begin{lemma}[Takenawa \cite{takenawa2}]\label{lem:degreeintersection}
\[
	\deg \varphi^{n} = (\Phi^n e^{(0)}, e^{(0)}).
\]
\end{lemma}
\begin{proof}
Consider the following commutative diagram:
\[
\xymatrix{
	\mathbb{Z}^{1, r} \ar[d]_{\iota_{\ell}} \ar[r]^{\Phi} & \cdots \ar[r]^{\Phi} & \mathbb{Z}^{1, r} \ar[d]^{\iota_{\ell+n}} \\
	\operatorname{Pic} X_{\ell} \ar[r]_{\varphi_{\ell*}} & \cdots \ar[r]_{\varphi_{\ell+n-1*}} & \operatorname{Pic} X_{\ell+n} \ar[d]^{\pi_{\ell+n*}} \\
	\operatorname{Pic}(\mathbb{P}^2) \ar[u]^{\pi^{*}_{\ell}} &  & \operatorname{Pic}(\mathbb{P}^2).
}
\]
Using Definition~\ref{app:p2degree}, we have
\begin{align*}
	\deg \varphi^n &= \left( \pi_{\ell+n*}(\varphi_{\ell+n-1} \circ \cdots \circ \varphi_{\ell})_{*} \pi^{*}_{\ell} \mathcal{O}_{\mathbb{P}^2}(1) \right) \cdot \mathcal{O}_{\mathbb{P}^2}(1) \\
	&= \left( (\varphi_{\ell+n-1} \circ \cdots \circ \varphi_{\ell})_{*} e^{(0)}_{\ell} \right) \cdot e^{(0)}_{\ell+n} \\
	&= \left( \iota_{\ell+n-1}\Phi^n e^{(0)} \right) \cdot e^{(0)}_{\ell+n} \\
	&= (\Phi^n e^{(0)}, e^{(0)}).
\end{align*}
\end{proof}

\begin{proposition}\label{prop:jordan}
The Jordan normal form of $\Phi$ is one of the following three:
\begin{itemize}
\item
\[
	\Phi
	\sim
	\left( \begin{matrix}
		\mu_1 & & \\
		& \ddots & \\
		& & \mu_{r+1}
	\end{matrix} \right),
\]
where $\mu_i$ are all roots of unity.
In particular, there exists $\ell > 0$ such that $\Phi^\ell = \operatorname{id}$ and thus the degree growth of the equation is bounded.

\item
\[
	\Phi
	\sim
	\left( \begin{matrix}
		1 & 1 & & & & & \\
		& 1 & 1 & & & & \\
		& & 1 & & & & \\
		& & & \mu_1 & & \\
		& & & & \ddots & \\
		& & & & & \mu_{r-2}
	\end{matrix} \right),
\]
where $\mu_i$ are all roots of unity.
In this case, the degree grows quadratically.
The dominant eigenvector is isotropic.

\item
\[
	\Phi
	\sim
	\left( \begin{matrix}
		\lambda & & & & & \\
		& \frac{1}{\lambda} & & & & \\
		& & \mu_1 & & \\
		& & & \ddots & \\
		& & & & \mu_{r-1}
	\end{matrix} \right),
\]
where $\lambda$ is a reciprocal quadratic integer greater than $1$ or a Salem number, and $|\mu_i| = 1$.
In this case, the entropy of the equation is $\log \lambda > 0$.
The two eigenvectors corresponding to $\lambda$ and $1/\lambda$ are both isotropic.
\end{itemize}
\end{proposition}

These three cases correspond to the classes $1, 3$ and $4$ in Theorem~\ref{thm:df}, respectively.

\begin{corollary}[Takenawa \cite{takenawa2}]\label{cor:entropyeigenvalue}
The dynamical degree of an equation is given by the maximum eigenvalue of $\Phi$ and the entropy by its logarithm.
\end{corollary}

\begin{corollary}
Theorem~\ref{thm:sicgrowth} still holds in the nonautonomous case.
\end{corollary}

\begin{remark}
We have already seen in Example~\ref{ex:nonpower} that Theorem~\ref{thm:autonomousgrowth} does {\em not} hold in general nonautonomous cases.
To extend Theorem~\ref{thm:autonomousgrowth} to the nonautonomous case, it is necessary to apply some conditions on the mapping $\varphi_n$ itself.
However, since there exist too many possible artificial equations in the nonautonomous case, it would be extremely difficult to describe such conditions in all generality.
\end{remark}

It is easy to prove Proposition~\ref{prop:jordan} if we admit the following two lemmas in linear algebra.

\begin{lemma}\label{lem:lem1}
Let $V$ be an $(r+1)$-dimensional $\mathbb{C}$-vector space with a Hermitian form $(-, -)$ of signature $(1, r)$.
If $v \in V$ is isotropic, i.e.\,$(v, v) = 0$ and $v \ne 0$, then the signature of $(-, -)\big|_{v^{\perp}}$ is $(0, r - 1)$ and its kernel is generated by $v$.
In particular, if $v_1, v_2$ satisfy $(v_1, v_1) = (v_1, v_2) = (v_2, v_2) = 0$, then $v_1$ and $v_2$ are linearly dependent.
\end{lemma}

\begin{lemma}\label{lem:symmetricjordan}
Let $V$ be an $(r+1)$-dimensional $\mathbb{R}$-vector space with a symmetric bilinear form $(-, -)$ of signature $(1, r)$, and let $f$ be a linear transformation on $V$ which preserves $(-, -)$.

\begin{itemize}
\item[(1)]
The Jordan normal form of $f$ must be one of the following:
\begin{equation}\label{eq:jordanbounded}
	\left( \begin{matrix}
		\mu_1 & & \\
		& \ddots & \\
		& & \mu_{r+1}
	\end{matrix} \right)
	\quad
	(|\mu_i| = 1),
\end{equation}
\begin{equation}\label{eq:jordanquadratic}
	\left( \begin{matrix}
		\nu & 1 & & & & & \\
		& \nu & 1 & & & & \\
		& & \nu & & & & \\
		& & & \mu_1 & & \\
		& & & & \ddots & \\
		& & & & & \mu_{r-2}
	\end{matrix} \right)
	\quad
	(\nu = \pm 1, |\mu_i| = 1),
\end{equation}
\begin{equation}\label{eq:jordanexp}
	\left( \begin{matrix}
		\lambda & & & & & \\
		& \frac{1}{\lambda} & & & & \\
		& & \mu_1 & & \\
		& & & \ddots & \\
		& & & & \mu_{r-1}
	\end{matrix} \right)
	\quad
	(\lambda \in \mathbb{R}, |\lambda| > 1, |\mu_i| = 1).
\end{equation}

\item[(2)]
Consider the case where the Jordan normal form of $f$ is (\ref{eq:jordanquadratic}).
If $(v_1, v_2, v_3, u_1, \ldots, u_{r-2})$ is the corresponding Jordan basis on $V_{\mathbb{C}} = V \otimes \mathbb{C}$, then $v_1$ is isotropic and
\[
	\lim_{n \to +\infty} \frac{1}{\nu^n n^2} f^n w = \frac{(w, v_1)}{2 (v_3, v_1)} v_1
\]
for any $w \in V_{\mathbb{C}}$.

\item[(3)]
Consider the case where the Jordan normal form of $f$ is (\ref{eq:jordanexp}).
If $(v_1, v_2, u_1, \ldots, u_{r-1})$ is the corresponding Jordan basis, then $v_1$ and $v_2$ are both isotropic and
\[
	\lim_{n \to +\infty} \frac{1}{\lambda^n} f^n w = \frac{(w, v_2)}{(v_1, v_2)} v_1
\]
for any $w \in V_{\mathbb{C}}$.
\end{itemize}
\end{lemma}

Although we shall use Lemma~\ref{lem:lem1} throughout the paper, we omit its proof since it is a well-known fact in linear algebra.
The proof of Lemma~\ref{lem:symmetricjordan} will be given in Appendix~\ref{sec:proofoflemma} since it is long and not often stated explicitly in the literature.

\begin{proof}[Proof of Proposition~\ref{prop:jordan}]

By Lemma~\ref{lem:symmetricjordan}, the Jordan normal form of $\Phi$ is (\ref{eq:jordanbounded}), (\ref{eq:jordanquadratic}) or (\ref{eq:jordanexp}).

\noindent
\underline{Case: (\ref{eq:jordanbounded})}

It is sufficient to show that every eigenvalue of $\Phi$ is a root of unity.
Since $\Phi$ preserves the lattice $\mathbb{Z}^{1, r}$, its characteristic polynomial has integer coefficients.
Since all roots of this polynomial have modulus $1$, they are all roots of unity by Kronecker's theorem \cite{rootofunity}.

\noindent
\underline{Case: (\ref{eq:jordanquadratic})}

It is clear that the degree growth is at most quadratic, and the reason why the $\mu_i$ are roots of unity is the same as above.
Therefore, it is sufficient to show that $\nu = 1$ and that the degree growth is actually quadratic.

Using Lemma~\ref{lem:degreeintersection} and Lemma~\ref{lem:symmetricjordan} (2), we have
\begin{align*}
	\lim_{n \to +\infty} \frac{\deg \varphi^n}{\nu^n n^2}
	&= \left( \lim_{n \to +\infty} \frac{1}{\nu^n n^2} \Phi^n e^{(0)}, e^{(0)} \right) \\
	&= \left( \frac{(e^{(0)}, v_1)}{2 (v_3, v_1)} v_1, e^{(0)} \right) \\
	&= \frac{(e^{(0)}, v_1)^2}{2 (v_3, v_1)}.
\end{align*}
Since $v_1$ is isotropic, $(e^{(0)}, v_1)$ is not $0$ and thus $\deg \varphi^n / \nu^n$ grows quadratically.
Since $\deg \varphi^n$ is always positive, we have $\nu = 1$.

\noindent
\underline{Case: (\ref{eq:jordanexp})}

Since $\lambda$ has modulus greater than $1$, as in the case of (\ref{eq:jordanquadratic}), we have
\[
	\lim_{n \to +\infty} \frac{\deg \varphi^n}{\lambda^n} = \frac{(e^{(0)}, v_1)(e^{(0)}, v_2)}{(v_1, v_2)}.
\]
We can prove $(e^{(0)}, v_1) \ne 0$, $(e^{(0)}, v_2) \ne 0$ and $\lambda > 1$ in the same way as above.
\end{proof}

The following proposition shows the relation between the Jordan normal form of $\Phi$ and the Picard number $\rho(X_n)$.

\begin{proposition}\label{prop:picardgrowthna}
\begin{itemize}
\item[(1)]
If $\rho(X_n) < 10$, then the degree growth of the equation is bounded.

\item[(2)]
If $\rho(X_n) \le 10$, then the degree growth of the equation is bounded or quadratic.

\end{itemize}
\end{proposition}
\begin{proof}
The key to the proof is that $\Phi$ preserves $K = 3e^{(0)} - e^{(1)} - \cdots - e^{(r)}$.

(1)
Since
\[
	(K, K) = K^2_{X_n} = 10 - \rho(X_n) > 0,
\]
the bilinear form $(-, -)$ is negative definite on $K^{\perp}$.
Since $\Phi \big|_{K^{\perp}}$ preserves the lattice $K^{\perp}$ with a negative definite bilinear form, there exists $\ell > 0$ such that $\left( \Phi \big|_{K^{\perp}} \right)^{\ell} = \operatorname{id}$.

Note that the above also implies that $K$ (and $-K$) is isotropic if and only if the Picard number of $X_n$ is $10$.

(2)
Let us assume that the dynamical degree $\lambda$ is greater than $1$ and show that $(K, K) < 0$.
Let $v \in \mathbb{R}^{1, r}$ be the eigenvector corresponding to $\lambda$.
Since
\[
	(v, K) = (\lambda v, K) = \lambda (v, K),
\]
we have $K \in v^{\perp}$.
Lemma~\ref{lem:lem1} says that $(-, -) \big|_{v^{\perp}}$ is semi-negative definite and its kernel is generated by $v$.
Since $v$ and $K$ are eigenvectors corresponding to different eigenvalues, we have $K \notin \mathbb{C}v$, and thus $(K, K) < 0$.
\end{proof}

Since all generalized Halphen surfaces have Picard number $10$, one immediately obtains the following corollary, which was first shown by Takenawa on a case-by-case basis \cite{takenawa2}.

\begin{corollary}
The degree growth of any discrete Painlev\'{e} equation is quadratic.
In particular, all discrete Painlev\'{e} equations are integrable.
\end{corollary}

As shown in the following example, the direct converse of Proposition~\ref{prop:picardgrowthna} does not hold, even in the autonomous case.

\begin{example}
Let $\varphi$ be an automorphism on a rational surface $X$ and let $P \in X$ be a fixed point of $\varphi$.
Let $\epsilon \colon \widetilde{X} \to X$ be the blow-up at $P$.
Then, $\varphi$ is lifted to an automorphism on $\widetilde{X}$.

This procedure does not change the algebraic entropy of an equation but increases the Picard number of a space of initial conditions.
\end{example}

When we consider the classification of equations with a space of initial conditions, it is sometimes necessary to perform a minimization of the space.
A concrete approach to such minimizations was considered by Carstea and Takenawa in \cite{ct}, where they gave an example of a minimization that contracts a curve passing through $\mathbb{C}^2$ (the finite region in $\mathbb{P}^2$ or $\mathbb{P}^1 \times \mathbb{P}^1$).
However, a general theory in the nonautonomous case was not yet known.
In the following section, we will consider a minimization of a space of initial conditions in the general case, in order to classify all nonautonomous integrable equations with unbounded degree growth that possess a space of initial conditions.


\section{Minimization of a space of initial conditions}\label{sec:minimalization}

Let us consider a minimization of a space of initial conditions for a nonautonomous mapping.
In this section, we consider the situation in Definition~\ref{def:sicsakai} and think of $(X, \sigma)$ as defining an equation.

We first recall the process of minimization in the autonomous case.

\begin{definition}\label{def:minimalautonomous}
Let $\varphi$ be an autonomous equation (automorphism) on a rational surface $X$.
Then, $\varphi$ is minimal on $X$ if there is no rational surface $X'$, no automorphism $\varphi'$ (on $X'$) and no birational morphism $\epsilon \colon X \to X'$, such that $\rho(X) > \rho(X')$ and $\epsilon \circ \varphi = \varphi' \circ \epsilon$:
\[
\xymatrix{
	X \ar[r]^{\sim}_{\varphi} \ar[d]^{\epsilon} & X \ar[d]^{\epsilon} \\
	X' \ar[r]^{\sim}_{\varphi'} & X'.
}
\]
\end{definition}

It is known that an automorphism $\varphi$ on $X$ is minimal if and only if there are no mutually disjoint exceptional curves of the first kind on $X$ that are permuted by $\varphi$.
The proof of this statement is almost exactly the same as that of Lemma~\ref{lem:minimalitycurves}.

Hence, what we call a minimization of an autonomous equation is first of all the process of finding such contractible curves and then to actually realize the contraction.

\begin{definition}\label{def:minimalna}
Let $X$ be a rational surface and let $\sigma$ be a Cremona isometry on $X$.
A nonautonomous equation $(X, \sigma)$ is minimal if there is no rational surface $X'$, no birational morphism $\epsilon \colon X \to X'$ and no Cremona isometry $\sigma'$ on $\operatorname{Pic} X'$, such that $\rho(X) > \rho(X')$ and $\epsilon_{*} \sigma = \sigma'  \epsilon_{*}$:
 \[
\xymatrix{
	\operatorname{Pic} X \ar[r]_{\sigma} \ar[d]^{\epsilon_{*}} & \operatorname{Pic} X \ar[d]^{\epsilon_{*}} \\
	\operatorname{Pic} X' \ar[r]_{\sigma'} & \operatorname{Pic} X'.
}
\]
\end{definition}

As in the autonomous case, it is possible to explicitly verify the minimality with specific types of curves.

\begin{lemma}\label{lem:minimalitycurves}
Let $X$ be a rational surface and $\sigma$ a Cremona isometry on $\operatorname{Pic} X$.
The equation $(X, \sigma)$ is minimal if and only if there are no mutually disjoint exceptional curves of the first kind $C_1, \ldots, C_N \subset X$ that are permuted by $\sigma$.
\end{lemma}

\begin{lemma}\label{lem:minimalcremona}
Let $X, X'$ be rational surfaces and $\epsilon \colon X \to X'$ a birational morphism.
If a Cremona isometry $\sigma$ on $\operatorname{Pic} X$ preserves the sublattice $\epsilon^{*}(\operatorname{Pic} X') \subset \operatorname{Pic} X$, then $\epsilon_{*} \sigma \epsilon^{*}$ is also a Cremona isometry on $\operatorname{Pic} X'$.
\end{lemma}

\begin{proof}
Let $\epsilon = \epsilon^{(1)} \circ \cdots \circ \epsilon^{(L)}$ be a decomposition into blow-ups and let $E^{(i)}$ be the total transform of the class of the exceptional curve of $\epsilon^{(i)}$ for $i = 1, \ldots, L$.

Let $F, F' \in \operatorname{Pic} X'$.
Using $\sigma \epsilon^{*} F, \sigma \epsilon^{*} F' \in \epsilon^{*}(\operatorname{Pic} X')$, we have
\begin{align*}
	(\epsilon_{*} \sigma \epsilon^{*} F) \cdot (\epsilon_{*} \sigma \epsilon^{*} F')
	&= (\sigma \epsilon^{*} F) \cdot (\sigma \epsilon^{*} F') \\
	&= F \cdot F'.
\end{align*}

Since $K_X = \epsilon^{*}K_{X'} + E^{(1)} + \cdots E^{(L)}$, $\epsilon_{*}E^{(i)} = 0$ and since the $E^{(i)}$ are permuted by $\sigma$ we have that
\begin{align*}
	\epsilon_{*} \sigma \epsilon^{*} K_{X'}
	&= \epsilon_{*} \sigma (K_X - E^{(1)} - \cdots E^{(L)}) \\
	&= \epsilon_{*} (K_X - E^{(1)} - \cdots E^{(L)}) \\
	&= K_{X'}.
\end{align*}

The third condition in Definition~\ref{def:cremona} is trivial since $\epsilon^{*}, \sigma$ and $\epsilon_{*}$ all preserve the effective class.
\end{proof}

\begin{proof}[Proof of Lemma~\ref{lem:minimalitycurves}]
First, let $C_1, \ldots, C_N$ be irreducible curves of the first kind that are permuted by $\sigma$.
It follows from Castelnuovo's contraction theorem that there exist a rational surface $X'$ and a birational morphism $\epsilon \colon X \to X'$ such that $\epsilon$ contracts $C_1, \ldots, C_N$ and is an isomorphism outside $C_1 \cup \cdots \cup C_N$.
Let $\sigma' = \epsilon_{*} \sigma \epsilon^{*}$.
By Lemma~\ref{lem:minimalcremona}, $\sigma'$ is a Cremona isometry on $\operatorname{Pic} X'$ and thus we obtain an equation $(X', \sigma')$.
It is clear, by construction, that $\epsilon, X', \sigma'$ satisfy the conditions in Definition~\ref{def:minimalna}.

Next we show the converse.
Let $\epsilon, X', \sigma'$ satisfy the conditions in Definition~\ref{def:minimalna}:
\[
\xymatrix{
	\operatorname{Pic} X \ar[r]_{\sigma} \ar[d]^{\epsilon_{*}} & \operatorname{Pic} X \ar[d]^{\epsilon_{*}} \\
	\operatorname{Pic} X' \ar[r]_{\sigma'} & \operatorname{Pic} X'
}
\]
and take an exceptional curve of the first kind $C$ that is contracted by $\epsilon$.
Since
\[
	\epsilon_{*} \sigma^{\ell}[C] = \sigma'^{\ell} \epsilon_{*}[C] = 0,
\]
$\sigma^{\ell}C$ is contracted by $\epsilon$ for all $\ell$.
However, $\epsilon$ contracts only a finite number of curves.
Thus, there exists $N > 0$ such that $\sigma^N C = C$.
Hence $\sigma$ acts as a permutation on $\{ C, \sigma C, \ldots, \sigma^{N-1}C \}$.
Since these curves are exceptional curves of the first kind and are contracted by $\sigma$, they are mutually disjoint.
\end{proof}

As in the autonomous case, one must first verify if there are such contractible curves.
If so, then we obtain an equation $(X', \sigma')$ by contracting these curves.
It is clear that the degree growth of $(X, \sigma)$ is the same as that of $(X', \sigma')$.
Replacing $(X, \sigma)$ with $(X', \sigma')$ and repeating this procedure, we finally obtain a surface on which the equation is minimal.

As shown in the following example, a minimization is not unique in general.

\begin{example}\label{ex:twominimalization}
Let $X$ be the surface obtained by blowing up $\mathbb{P}^1 \times \mathbb{P}^1$ at $(\infty, \infty)$, and let $\varphi(x, y) = (y, x)$.
It is clear that $\varphi$ is an automorphism on $X$.

$X$ has three exceptional curves of the first kind: $C, \{ x = \infty \}$ and $\{ y = \infty \}$ (Figure~\ref{fig:twominimalization}).
This mapping has two minimizations.

The first possibility is $\mathbb{P}^1 \times \mathbb{P}^1$.
Since $C$ is fixed by $\varphi$, we can minimize $\varphi$ from $X$ to $\mathbb{P}^1 \times \mathbb{P}^1$, and it is trivial that $\varphi$ is an automorphism on $\mathbb{P}^1 \times \mathbb{P}^1$.

The second possibility is $\mathbb{P}^2$.
Since two curves $\{ y = \infty \}$ and $\{ x = \infty \}$ are permuted by $\varphi$, we can minimize $\varphi$ from $X$ to $\mathbb{P}^2$ by contracting these curves.
\begin{figure}
\begin{picture}(320, 280)

	{\thicklines
	\put(10, 30){\line(1, 0){100}}
	\put(30, 10){\line(0, 1){100}}
	\put(10, 90){\line(1, 0){100}}
	\put(90, 10){\line(0, 1){100}}
	}
	
	\put(90, 90){\circle*{5}}

	\put(90, 130){$\swarrow$}

	{\thicklines
	\put(110, 180){\line(1, 0){100}}
	\put(130, 160){\line(0, 1){100}}
	\put(110, 240){\line(1, 0){50}}
	\put(190, 160){\line(0, 1){50}}
	
	\qbezier(160, 240)(180, 240)(190, 260)
	\qbezier(190, 210)(190, 230)(210, 240)
	
	}
	
	\put(160, 260){\line(1, -1){50}}

	\put(60, 237){$\{ y = \infty \}$}
	\put(215, 240){$\{ x = \infty \}$}
	\put(215, 205){$C$}

	{\thicklines
	\put(210, 30){\line(1, 0){100}}
	\put(230, 10){\line(0, 1){100}}
	\put(210, 110){\line(1, -1){100}}
	}
	
	\put(230, 90){\circle*{5}}
	\put(290, 30){\circle*{5}}
	
	\put(210, 130){$\searrow$}

\end{picture}
\caption{
The mapping in Example~\ref{ex:twominimalization} permutes two axes $x$ and $y$.
If we consider this mapping on the upper surface, it has two minimizations.
}
\label{fig:twominimalization}
\end{figure}
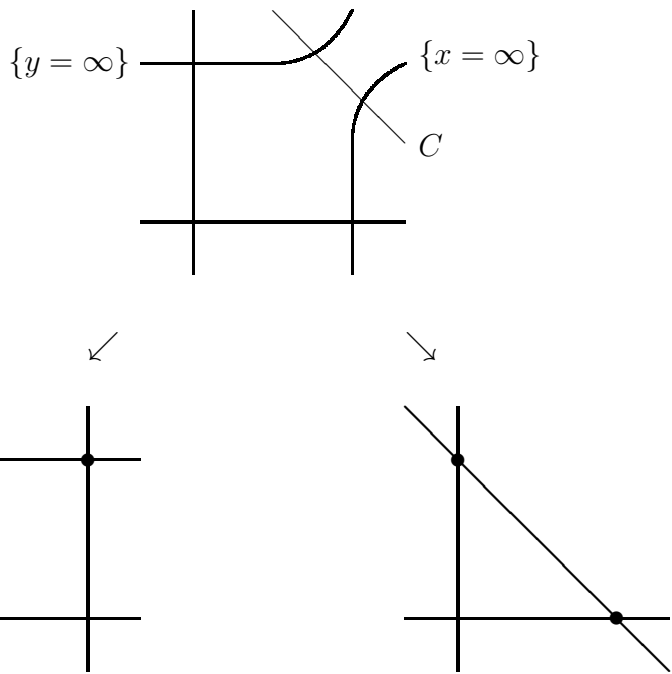
\end{example}

We will show in Proposition~\ref{prop:unique} (for the integrable case) and Proposition~\ref{prop:expunique} (in the nonintegrable case) that if the degree growth is unbounded, i.e.\,$\sigma$ is of infinite order, then the minimization is unique.


\subsection{Integrable case}\label{subsec:main}

In this subsection, we consider a minimization in the case of integrable equations.

In the autonomous case, a minimal space of initial conditions is always an elliptic surface and the equation preserves the elliptic fibration \cite{rationalgsurfaces}.
Thus, the theory of rational elliptic surfaces is relevant to the problem of minimization.
In the nonautonomous case, however, a space of initial conditions does not have an elliptic fibration in general \cite{sakai}.
Therefore, we will have to find the contractible curves in Lemma~\ref{lem:minimalitycurves} in a different way.

The following is our main theorem in this paper:

\begin{theorem}\label{thm:main}
Consider an equation $(X, \sigma)$ and assume that its degree growth is quadratic.
Then, this equation can be minimized to a generalized Halphen surface.

In particular, if a mapping of the plane with unbounded degree growth and zero algebraic entropy has a space of initial conditions, then it must be one of the discrete Painlev\'{e} equations.
\end{theorem}

Note that in this paper, ``discrete Painlev\'{e} equation'' should be understood in Sakai's sense, as defined in Definition~\ref{def:discretepainleve}.

\begin{lemma}\label{lem:nefhalphen}
Let $X$ be a rational surface with $\rho(X) = 10$.
Then $X$ is a generalized Halphen surface if and only if $-K_X$ is nef.
\end{lemma}
\begin{proof}
Suppose $X$ is a generalized Halphen surface.
Let $C \subset X$ be an irreducible curve.
If $C$ is a component of $-K_X$, then $- K_X \cdot C = 0$, by definition.
On the other hand, if $C$ is not a component of $-K_X$, then $- K_X \cdot C \ge 0$ since $-K_X$ is effective.
In both cases we have $-K_X \cdot C \ge 0$ and thus $-K_X$ is nef.

Let us prove the converse.
Suppose $-K_X$ is nef.
Since $\rho(X) = 10$, we have $(-K_X)^2 = 0$ and $-K_X$ is effective (\cite{sakai}, Proposition~2).
Thus it is sufficient to show that every component of $-K_X$ is orthogonal to $-K_X$.
Let $\sum_i a_i D_i \in |- K_X|$.
Since $-K_X$ is nef, we have $a_i D_i \cdot (- K_X) \ge 0$.
Summing, we obtain
\[
	\sum_i a_i D_i \cdot (- K_X) \ge 0.
\]
Since the left hand side is equal to $(-K_X)^2$, $D_i \cdot (- K_X)$ must be $0$ for all $i$.
Hence $X$ is a generalized Halphen surface.
\end{proof}

\begin{lemma}\label{lem:v1nef}
Let $X$ be a basic rational surface and let $\sigma$ be a Cremona isometry on $\operatorname{Pic} X$ with quadratic growth.
Let $v_1, v_2, v_3 \in \operatorname{Pic}_{\mathbb{Q}} X \setminus \{ 0 \}$ satisfy
\begin{align*}
	\sigma v_1 &= v_1, \\
	\sigma v_2 &= v_2 + v_1, \\
	\sigma v_3 &= v_3 + v_2.
\end{align*}
Then, we have
\begin{itemize}
\item
$v_1$ is isotropic,

\item
either $v_1$ or $- v_1$ is nef,

\item
$v_1 \cdot K_X = 0$.
\end{itemize}
\end{lemma}

\begin{proof}
That $v_1$ is isotropic follows immediately from Proposition~\ref{prop:jordan} .

Let $e = (e^{(0)}, \ldots, e^{(r)})$ be a geometric basis.
Then, by Proposition~\ref{prop:jordan} and Lemma~\ref{lem:symmetricjordan}, there exists $a \in \mathbb{Q}^{\times}$ such that
\[
	v_1 = a \lim_{n \to +\infty} \frac{1}{n^2} \sigma^n e^{(0)}.
\]
Since $e^{(0)}$ is nef and $\sigma$ preserves the nef cone (Lemma~\ref{lem:cremonacone}), $\frac{1}{n^2} \sigma^n e^{(0)}$ is nef for all $n$.
Therefore, Proposition~\ref{app:nefclosed} implies that $\frac{1}{a}v_1$ is nef.

Since
\begin{align*}
	v_2 \cdot K_X &= (\sigma v_2) \cdot (\sigma K_X) \\
	&= (v_2 + v_1) \cdot K_X \\
	&= v_2 \cdot K_X + v_1 \cdot K_X,
\end{align*}
we have $v_1 \cdot K_X = 0$.
\end{proof}

Note that while $v_2, v_3$ above are not unique, $v_1$ is unique up to scaling.
$v_1$ is determined by
\[
	\mathbb{Q}v_1 = \operatorname{Ker}(\sigma_{\mathbb{Q}} - \operatorname{id}) \cap \operatorname{Im}(\sigma_{\mathbb{Q}} - \operatorname{id})^2,
\]
where $\sigma_{\mathbb{Q}}$ is the $\mathbb{Q}$-extension of $\sigma$ to $\operatorname{Pic}_{\mathbb{Q}} X$.

\begin{definition}
Let us normalize $v_1$ so that
\begin{itemize}
\item
$v_1$ is nef,
\item
$v_1 \in \operatorname{Pic} X$,
\item
$v_1$ is primitive in $\operatorname{Pic} X$, i.e.\,if a rational number $a$ satisfies $a v_1 \in \operatorname{Pic} X$, then $a$ is an integer.
\end{itemize}
We shall call this $v_1$ the {\em normalized dominant eigenvector} of $\sigma$.
\end{definition}

\begin{lemma}\label{lem:pic10halphen}
Let $X$ be a rational surface of Picard number $10$.
If $X$ has a Cremona isometry which grows quadratically, then $X$ must be a generalized Halphen surface and $-K_X$ coincides with the normalized dominant eigenvector.
\end{lemma}

\begin{proof}
Let $\sigma$ be a Cremona isometry on $\operatorname{Pic} X$ that grows quadratically and let $v_1$ be the normalized dominant eigenvector of $\sigma$.
By Lemma~\ref{lem:v1nef}, $v_1$ is isotropic and $v_1 \cdot K_X = 0$.
However, $K_X$ is also isotropic since $\rho(X) = 10$.
Therefore, by Lemma~\ref{lem:lem1}, $v_1$ and $K_X$ are linearly dependent.
Since $v_1$ and $K_X$ are both primitive in $\operatorname{Pic} X$, we have $v_1 = \pm K_X$.
While $v_1$ is nef by Lemma~\ref{lem:v1nef}, $K_X$ cannot be nef since $X$ is rational.
Thus we have $v_1 = -K_X$ and Lemma~\ref{lem:nefhalphen} implies that $X$ is a generalized Halphen surface.
\end{proof}

The following lemma is the key to the proof of Theorem~\ref{thm:main}

\begin{lemma}\label{lem:keylemma}
Let $X$ be a rational surface with $\rho(X) > 10$ and let $\sigma$ be a Cremona isometry on $\operatorname{Pic} X$ with quadratic growth.
Then, $X$ is not minimal for the equation $(X, \sigma)$.
\end{lemma}

\begin{proof}
We will try to find mutually disjoint exceptional curves of the first kind that are permuted by $\sigma$ (Lemma~\ref{lem:minimalitycurves}).

\noindent
\underline{Step~1}

Let $v_1 \in \operatorname{Pic} X$ be the normalized dominant eigenvector of $\sigma$.
We first show that $v_1 + K_X$ is effective and nonzero.

By the Riemann-Roch inequality, we have
\[
	h^0(v_1 + K_X) + h^2(v_1 + K_X) \ge 1 + \frac{1}{2} (v_1 + K_X) \cdot v_1 = 1.
\]
Using Serre duality we have
\[
	h^2(v_1 + K_X) = h^0(-v_1) = 0.
\]
Hence, $h^0(v_1 + K_X) \ge 1$ and $v_1 + K_X$ is effective.
It immediately follows from $(v_1 + K_X)^2 = 10 - \rho(X) < 0$ that $v_1 + K_X \ne 0$.

\noindent
\underline{Step~2}

Let
\[
	\mathcal{C} = \{ C \subset X \colon \text{irreducible} \, | \, C \cdot (v_1 + K_X) < 0 \}.
\]
We show that $\mathcal{C}$ is a nonempty finite set.

By Step~1, we can express $v_1 + K_X$ as
\[
	v_1 + K_X = \left[ \sum^{\ell}_{i=1} a_i C_i \right],
\]
where the $C_i$ are irreducible and $a_i > 0$.
Since $(v_1 + K_X)^2 < 0$, at least one of $C_1, \ldots, C_{\ell}$ satisfy $C_i \cdot (v_1 + K_X) < 0$.
Thus $\mathcal{C}$ is not empty.

On the other hand, if an irreducible curve $C$ is different from $C_1, \ldots, C_{\ell}$, then it satisfies $C \cdot (v_1 + K_X) \ge 0$.
Hence $\mathcal{C}$ is finite.

\noindent
\underline{Step~3}

We show that if $C \in \mathcal{C}$, then
\[
	C^2 = -1, \quad
	C \cdot K_X = -1, \quad
	C \cdot v_1 = 0, \quad
	C \cong \mathbb{P}^1.
\]

By the genus formula, we have
\begin{align*}
	g_a(C) &= 1 + \frac{1}{2} C \cdot (C + K_X) \\
	&= 1 + \frac{1}{2} C^2 + \frac{1}{2} C \cdot (v_1 + K_X) - \frac{1}{2} C \cdot v_1.
\end{align*}
Since $g_a(C) \ge 0$, $C^2 < 0$, $C \cdot (v_1 + K_X) < 0$ and $C \cdot v_1 \ge 0$, the only possible case is
\[
	g_a(C) = 0, \quad
	C^2 = -1, \quad
	C \cdot (v_1 + K_X) = -1, \quad
	C \cdot v_1 = 0.
\]
It follows from Proposition~\ref{app:p1arith} that $C \cong \mathbb{P}^1$.

\noindent
\underline{Step~4}

Since $\sigma$ is a Cremona isometry, Lemma \ref{lem:crirred} implies that $\sigma$ acts on $\mathcal{C}$ as a permutation.

\noindent
\underline{Step~5}

Let $C, C' \in \mathcal{C}$ satisfy $C \ne C$.
We show that $C \cap C' = \emptyset$.

Let $m = C \cdot C'$.
Since $(C + C') \cdot v_1 = 0$, Lemma~\ref{lem:lem1} implies that
\[
	0 \ge (C + C')^2 = 2m - 2
\]
and therefore $m = 0$ or $m = 1$.
Assume that $m = 1$.
In this case, $v_1$ and $C + C'$ are orthogonal and both isotropic.
Thus, again by Lemma~\ref{lem:lem1}, there exists $a \in \mathbb{Q}^{\times}$ such that $[C + C'] = a v_1$.
Since $v_1$ and $[C + C']$ are both primitive and effective, we have $a = 1$.
On the other hand, since $C$ and $C'$ are two different components of $v_1 + K_X$, there exists an effective class $F$ such that $[C] + [C'] + F = v_1 + K_X$.
Thus we have $F = K_X$, which is a contradiction since $K_X$ cannot be effective when $X$ is rational.
Hence we have $C \cdot C' = 0$.
\end{proof}

\begin{proof}[Proof of Theorem~\ref{thm:main}]
Let $X$ be a rational surface and let $\sigma$ be a Cremona isometry on $\operatorname{Pic} X$ with quadratic growth.
We show that one can minimize $\sigma$ from $X$ to a generalized Halphen surface.

It follows from Proposition~\ref{prop:picardgrowthna} that $\rho(X) \ge 10$.
If $\rho(X) = 10$, then Lemma~\ref{lem:pic10halphen} implies that $X$ is a generalized Halphen surface, and thus $(X, \sigma)$ is a discrete Painlev\'{e} equation.

Consider the case $\rho(X) > 10$.
By Lemma~\ref{lem:keylemma}, the equation $(X, \sigma)$ is not minimal.
Let $\epsilon \colon X \to X'$ be a minimization and let $\sigma' = \epsilon_{*} \sigma \epsilon^{*}$.
The minimality of $(X', \sigma')$ implies that $\rho(X') \le 10$.
However, it follows from Proposition~\ref{prop:picardgrowthna} that $\rho(X') \ge 10$ since the degree grows quadratically.
Thus Lemma~\ref{lem:pic10halphen} implies that $X$ is a generalized Halphen surface and hence the equation $(X', \sigma')$ is a discrete Painlev\'{e} equation.
\end{proof}

Although the proofs of Lemma~\ref{lem:keylemma} and Theorem~\ref{thm:main} define a program to minimize $(X, \sigma)$, it could be a little difficult to describe the set $\mathcal{C}$ in Step~2 of the proof of Lemma~\ref{lem:keylemma} explicitly.
The following proposition tells us how to find a minimization only by linear algebra and, at the same time, shows the uniqueness of the minimization.

\begin{proposition}\label{prop:unique}
Let $X$ be a rational surface with $\rho(X) = r + 1 > 10$ and $\sigma$ a Cremona isometry on $\operatorname{Pic} X$ that grows quadratically.
Let $\epsilon \colon X \to X'$ be a minimization of $(X, \sigma)$.
Decompose $\epsilon$ into a composition of blow-ups
\[
	\epsilon = \epsilon^{(1)} \circ \cdots \circ \epsilon^{(r-9)}
\]
and let $E^{(i)} \in \operatorname{Pic} X$ be the total transform of the exceptional class of $\epsilon^{(i)}$ for $i = 1, \ldots, r - 9$.
Let $v_1 \in \operatorname{Pic} X$ be the normalized dominant eigenvector and $e = (e^{(0)}, \ldots, e^{(r)})$ an arbitrary geometric basis on $\operatorname{Pic} X$.
Then the set $\{ E^{(1)}, \ldots E^{(r-9)} \}$ can be written as
\begin{equation}\label{eq:excep}
	\mathcal{E} = \{ E \in \operatorname{Pic} X \, | \, E^2 = -1, E \cdot v_1 = 0, E \cdot K_X = -1, E \cdot e^{(0)} \ge 0, (v_1 - E) \cdot e^{(0)} \ge 3 \}.
\end{equation}
In particular, a minimization $\epsilon \colon X \to X'$ is unique.
\end{proposition}

\begin{proof}
\noindent
\underline{Step~1}

We show that $E^{(i)} \in \mathcal{E}$ for $i = 1, \ldots, r - 9$.
It is sufficient to show that $(v_1 - E^{(i)}) \cdot e^{(0)} \ge 3$ since the other conditions are trivial.
Since $e^{(0)}$ is nef and
\[
	v_1 + K_X - E^{(i)} = E^{(1)} + \cdots E^{(i-1)} + E^{(i+1)} + \cdots E^{(r-9)}
\]
is effective, we have
\begin{align*}
	0 & \le (v_1 + K_X - E^{(i)}) \cdot e^{(0)} \\
	&= - 3 + (v_1 - E^{(i)}) \cdot e^{(0)}.
\end{align*}

\noindent
\underline{Step~2}

Let $E \in \mathcal{E}$.
We show that $E$ and $K_X + v_1 - E$ are both effective.

By the Riemann-Roch inequality, we have
\[
	h^0(E) + h^2(E) \ge 1 + E \cdot (E - K_X) = 1.
\]
Using Serre duality, we have
\[
	h^2(E) = h^0(K_X - E).
\]
Since
\[
	(K_X - E) \cdot e^{(0)} \le - 3,
\]
$K_X - E$ is not effective and thus $h^0(K_X - E) = 0$.
Therefore we have $h^0(E) > 0$.

By the Riemann-Roch inequality and Serre duality, we have
\begin{align*}
	h^0(K_X + v_1 - E) & \ge 1 + (K_X + v_1 - E) \cdot (v_1 - E) - h^2(K_X + v_1 - E) \\
	& = 1 - h^0(-v_1 + E).
\end{align*}
It follows from $(-v_1 + E) \cdot e^{(0)} < 0$ that $h^0(-v_1 + E) = 0$.
Thus we have $h^0(K_X + v_1 - E) > 0$.

\noindent
\underline{Step~3}

We show that if $E, E' \in \mathcal{E}$ and $E \ne E'$, then $E \cdot E' = 0$.
It is important to note that $E, E' \in v^{\perp}_1$ and that the intersection is semi-negative definite on $v^{\perp}_1$ and its kernel is generated by $v_1$.

Let $m = E \cdot E'$.
Since
\[
	0 \ge (E \pm E')^2 = - 2 \pm 2m,
\]
we have $m = 0, \pm 1$.
We can exclude the cases $m = \pm 1$ as follows.

Assume that $m = 1$.
In this case, $E + E'$ is isotropic and thus there exists $\alpha$ such that $E + E' = \alpha v_1$.
However, this leads to the contradiction:
\[
	0 = \alpha v_1 \cdot K_X = (E + E') \cdot K_X = -2.
\]

Assume that $m = -1$.
As in the case of $m = 1$, there exists $\alpha$ such that $E - E' = \alpha v_1$.
Since $v_1$ is primitive, $\alpha$ is a nonzero integer.
We may assume $\alpha > 0$.
Thus we have
\[
	E' + (K_X + v_1 - E) = (1 - \alpha)v_1 + K_X.
\]
However, while the left hand side is effective, the right hand side is not.
Hence we conclude that $E \cdot E' = 0$.

\noindent
\underline{Step~4}

We show that $\mathcal{E} \subset \{ E^{(1)}, \ldots, E^{(9-r)} \}$.

Assume that there exists $E \in \mathcal{E} \setminus \{ E^{(1)}, \ldots, E^{(9-r)} \}$.
It follows from Steps 1 and 3 that $E \cdot E^{(i)} = 0$ for $i = 1, \ldots, E^{(r-9)}$.
However, this leads to the contradiction:
\[
	-1 = E \cdot K_X = E \cdot (-v_1 - E^{(1)} - \cdots - E^{(r-9)}) = 0.
\]

\noindent
\underline{Step~5}

The uniqueness of the minimization follows from the fact that the set $\mathcal{E}$ does not depend on $\epsilon$.
\end{proof}

The normalized dominant eigenvector $v_1$ is determined by
\[
	\mathbb{Z} v_1 = \operatorname{Ker}(\sigma_{\mathbb{Q}} - \operatorname{id}) \cap \operatorname{Im}(\sigma_{\mathbb{Q}} - \operatorname{id})^2 \cap \operatorname{Pic} X \quad
\text{ and } \quad
	v_1 \cdot e^{(0)} > 0.
\]
Thus, in principle, we can calculate $v_1$ and therefore $\mathcal{E}$ explicitly.
Hence, this proposition allows us to obtain $(X', \sigma')$ from $(X, \sigma)$ by mere linear algebra.

\begin{example}
Let us consider the equation in Example~\ref{ex:bupnormal}.

We have already constructed a space of initial conditions $X_n$ for this equation.
While the equation is integrable, $X_n$ has Picard number $11$.
Therefore, this space of initial conditions is not minimal and should have one contractible curve.
Let us explicitly find the contractible curve.
Instead of introducing $\iota_n$ and $\Phi$ as in Remark~\ref{rem:phi}, we identify all $\operatorname{Pic} X_n$ by using the basis $D^{(1)}_n, D^{(2)}_n, C^{(1)}_n, \ldots, C^{(6)}_n, \widetilde{C}^{(1)}_n, \widetilde{C}^{(2)}_n, \widetilde{C}^{(3)}_n$ for all $n$.

First we calculate the class $v_1 + K_{X_n}$, where $v_1 \in \operatorname{Pic} X_n$ is the normalized dominant eigenvector.
Using the matrix (\ref{eq:exmatrix}), we find that $v_1$ is written as
\[
	v_1 = 3 D^{(1)} + 3 D^{(2)} + 2 C^{(1)} + \cdots + 2 C^{(6)} - 2 \widetilde{C}^{(1)} - 2 \widetilde{C}^{(2)} - 2 \widetilde{C}^{(3)}.
\]
Let $H_x, H_y \in \operatorname{Pic} X_n$ be the total transforms of the curves $\{ x = \text{const} \}$ and $\{ y = \text{const} \}$, respectively.
Then, these classes can be written as
\[
	H_x = D^{(1)} + C^{(2)} + C^{(4)} + C^{(6)}, \quad
	H_y = D^{(2)} + C^{(1)} + C^{(3)} + C^{(5)},
\]
and thus we have
\[
	v_1 = 3 H_x + 3 H_y - C^{(1)} - \cdots - C^{(6)} - 2 \widetilde{C}^{(1)} - 2 \widetilde{C}^{(2)} - 2 \widetilde{C}^{(3)}.
\]
Since
\[
	K_{X_n} = -2 D^{(1)} -2 D^{(2)} - C^{(1)} - \cdots - C^{(6)} + \widetilde{C}^{(1)} + \widetilde{C}^{(2)} + \widetilde{C}^{(3)},
\]
we have
\[
	v_1 + K_X = H_x + H_y - \widetilde{C}^{(1)} - \widetilde{C}^{(2)} - \widetilde{C}^{(3)}.
\]

Next, let us find the contractible curve $C \subset X_n$.
Since $C$ is the only contractible curve in $X_n$, the class $v_1 + K_{X_n}$ must represent $C$.
Thus, the above expression implies that the image of $C$ in $\mathbb{P}^1 \times \mathbb{P}^1$ has bidegree $(1, 1)$ and passes through $Q^{(1)}, Q^{(3)}, Q^{(3)}$.
A direct calculation shows that the image of $C$ in $\mathbb{P}^1 \times \mathbb{P}^1$ is
\[
	\{ y_n - x_n = 2 a_n - \alpha \}.
\]
In fact, one can find that
\[
	\varphi_n \left( \{y_n - x_n = 2 a_n - \alpha \} \right) = \{y_{n+1} - x_{n+1} = 2 a_{n+1} - \alpha \},
\]
which means that, in a sense, the curve $\{ y - x = 2 a - \alpha \}$ is invariant under the equation.
Contracting this curve, we obtain the minimal space of initial conditions.
\end{example}


\subsection{Nonintegrable case}\label{subsec:exp}

In this subsection, we consider a minimization of a space of initial conditions in the nonintegrable case.
We will give in Proposition~\ref{prop:expminimal} a minimality criterion for a space of initial conditions and we show in Proposition~\ref{prop:expunique} that the minimization of a space of initial conditions is unique.

In this subsection, we consider the following situation:
\begin{itemize}
\item
$X$: a basic rational surface with $\rho(X) = r + 1 > 10$.
\item
$\sigma$: a Cremona isometry on $\operatorname{Pic} X$ with exponential growth.
\item
$\lambda > 1$: the maximum eigenvalue of $\sigma$.
\item
$v \in \operatorname{Pic}_{\mathbb{R}} X$: the dominant eigenvector of $\sigma$, which is isotropic.
\end{itemize}

\begin{lemma}
$v$ or $-v$ is nef.
\end{lemma}

\begin{proof}
The proof is the same as that of Lemma~\ref{lem:v1nef}.
Take a geometric basis $e = (e^{(0)}, \ldots, e^{(r)})$ and consider the limit
\[
	\lim_{n \to +\infty}\frac{1}{\lambda^n}\sigma^n e^{(0)}.
\]
\end{proof}

As the sign can be changed at will, we may assume that $v$ is nef.

We shall use the following lemma throughout this subsection.

\begin{lemma}\label{lem:negativedefinite}
The intersection number is negative definite on the lattice $v^{\perp} \cap \operatorname{Pic} X$.
\end{lemma}

\begin{proof}
Lemma~\ref{lem:lem1} says that the intersection number has signature $(0, r - 1)$ on $v^{\perp}$ and its kernel is generated by $v$.
However, a scalar multiple of $v$ does not belong to $\operatorname{Pic} X$ since $\lambda$ is irrational.
Thus the intersection number is negative definite on $v^{\perp} \cap \operatorname{Pic} X$.
\end{proof}

\begin{lemma}\label{lem:expdisjoint}
Two different exceptional curves of the first kind that belong to $v^{\perp}$ are always orthogonal to each other.
\end{lemma}

\begin{proof}
Let $C, C' \in v^{\perp}$ be two different exceptional curves of the first kind.
Using Lemma~\ref{lem:negativedefinite}, we have
\[
	0 > (C + C')^2 = -2 + 2 C \cdot C'
\]
and thus $C \cdot C' = 0$.
\end{proof}

\begin{proposition}\label{prop:expminimal}
$(X, \sigma)$ is minimal if and only if there exist no exceptional curves of the first kind that are orthogonal to $v$.
\end{proposition}

\begin{proof}
Suppose that $(X, \sigma)$ is not minimal.
Then there exist mutually disjoint exceptional curves of the first kind $C_1, \ldots, C_N$ such that $\sigma$ acts as a permutation on $\{ C_1, \ldots, C_N \}$.
It is therefore sufficient to show that $C_1 \cdot v = 0$.

Taking $\ell \in \mathbb{Z}$ such that $\sigma^{\ell}C_1 = C_1$, we have
\[
	C_1 \cdot v = (\sigma^{\ell}C_1) \cdot (\sigma^{\ell} v) = \lambda^{\ell} C_1 \cdot v,
\]
which shows that $C_1 \cdot v = 0$.

We now show the converse.
Let $\mathcal{C}$ be the set of the exceptional curves of the first kind that are orthogonal to $v$, and assume that $\mathcal{C}$ is nonempty.
It is clear that $\sigma$ acts on $\mathcal{C}$ as a permutation.
Since all elements in $\mathcal{C}$ are mutually disjoint by Lemma~\ref{lem:expdisjoint}, it follows from Lemma~\ref{lem:minimalitycurves} that $(X, \sigma)$ is not minimal.
\end{proof}

While the Picard numbers of the minimal spaces of initial conditions for integrable systems are always $10$, those of nonintegrable systems depend on the detail of the equations.
Therefore, it is impossible to check the minimality only by the Picard number.
We can only say that the Picard numbers are greater than $10$.
However, Proposition~\ref{prop:expminimal} gives us a precise minimality criterion.

The following proposition is an analogue of Proposition~\ref{prop:unique}.

\begin{proposition}\label{prop:expunique}
Let $\epsilon \colon X \to X'$ be a minimization of $(X, \sigma)$.
Decompose $\epsilon$ into a composition of blow-ups
\[
	\epsilon = \epsilon^{(1)} \circ \cdots \circ \epsilon^{(L)}
\]
and let $E^{(i)} \in \operatorname{Pic} X$ be the total transform of the exceptional class of $\epsilon^{(i)}$ for $i = 1, \ldots, L$.
Let $e = (e^{(0)}, \ldots, e^{(r)})$ be an arbitrary geometric basis on $\operatorname{Pic} X$.
Then the set $\{ E^{(1)}, \ldots E^{(L)} \}$ can be written as
\begin{equation*}
	\mathcal{E} = \{ E \in \operatorname{Pic} X \, | \, E^2 = -1, E \cdot v = 0, E \cdot K_X = -1, E \cdot e^{(0)} \ge 0 \}.
\end{equation*}
In particular, the minimization $\epsilon \colon X \to X'$ is unique.
\end{proposition}

\begin{proof}
It is clear that $E_i \in \mathcal{E}$ for $i = 1, \ldots, L$.
We show that $\mathcal{E} \subset \{ E_1, \ldots, E_L \}$.

\noindent
\underline{Step~1}

We show that every element $E \in \mathcal{E}$ is effective.
Using the Riemann-Roch inequality and Serre duality, we have
\[
	h^0(E) \ge 1 - h^2(E) = 1 - h^0(K_X - E).
\]
It follows from $e^{(0)} \cdot (K_X - E) < 0$ that $h^0(K_X - E) = 0$.
Thus $E$ is effective.

\noindent
\underline{Step~2}

We show that two different elements $E, E' \in \mathcal{E}$ are orthogonal to each other.
Since the intersection number is negative definite on $v^{\perp} \cap \operatorname{Pic} X$, we have
\[
	0 > (E \pm E')^2 = -2 \pm 2 E \cdot E',
\]
and thus $E \cdot E' = 0$.

\noindent
\underline{Step~3}

Assume that there exists $E \in \mathcal{E} \setminus \{ E_1, \ldots, E_L \}$.
Let $E' = \epsilon_{*}E$ and $v' = \epsilon_{*}v$.
Since $E \cdot E_i = 0$ by Step~2, we have
\[
	E'^2 = -1, \quad
	E' \cdot K_{X'} = -1, \quad
	E' \cdot v' = 0.
\]
Since $E$ is effective, so is $E'$.
Let us express $E'$ as a sum of irreducible curves:
\[
	E' = \sum^{\ell}_{j=1} a_i [C_i].
\]

We show that there exists at least one $j$ such that $C_j$ is an exceptional curve of the first kind.
Since $E' \cdot v' = 0$ and $v'$ is nef, we have $C_j \cdot v' = 0$ for $j = 1, \ldots, \ell$.
Since the intersection number is negative definite on $v'^{\perp} \cap \operatorname{Pic} X'$, $C^2_j$ are all negative.
By the genus formula, we have
\[
	g_a(C_j) = 1 + \frac{1}{2}C^2_j + \frac{1}{2}C_j \cdot K_{X'}.
\]
Multiplying with $a_j$ and summing, we obtain
\[
	\sum_j a_j g_a(C_j) = \sum_j a_j + \frac{1}{2} \sum_j a_j C^2_j + \frac{1}{2}\sum_j a_j C_j \cdot K_{X'}.
\]
Using $\sum_j a_j C_j = E'$ and $E' \cdot K_{X'} = -1$, we have
\[
	\sum_j a_j(2 + C^2_j - 2 g_a(C_j)) = 1.
\]
If $C^2_j - 2 g_a(C_j) \le -2$ for all $j$, then the left hand side is not positive, which is a contradiction.
Therefore, there is at least one $j$ such that $C^2_j - 2 g_a(C_j) \ge -1$.
Since $C^2_j < 0$, the only possible case is
\[
	C^2_j = -1, \quad
	g_a(C_j) = 0.
\]
Thus $C_j$ is an exceptional curve of the first kind.

Since $C_j \cdot v' = 0$, Proposition~\ref{prop:expminimal} implies that $(X', \sigma')$ is not minimal, which is a contradiction.
Hence, we have $\mathcal{E} = \{ E_1, \ldots, E_L \}$.
\end{proof}


\section{Conclusion}

In this paper, we studied nonautonomous mappings of the plane with spaces of initial conditions and unbounded degree growth by means of spaces of initial conditions.
Especially, Theorem~\ref{thm:main} shows that if an integrable mapping of the plane with unbounded degree growth has a space of initial conditions, then it must be one of the discrete Painlev\'{e} equations.
Since all discrete Painlev\'{e} equations have already been classified by Sakai \cite{sakai}, this means we have finished the classification of integrable mappings of the plane with a space of initial conditions and unbounded degree growth.
Moreover, we have given a concrete procedure to minimize a space of initial conditions to a generalized Halphen surface in the integrable case, as well as a general procedure to minimize the space of initial conditions for nonintegrable mappings.


\section*{Acknowledgement}

I wish to express my gratitude to Professors R. Willox, A. Ramani and B. Grammaticos for many useful discussions and comments.
I would also like to thank A. Ramani for providing me with the nice examples in Examples~\ref{ex:firstexample} and \ref{ex:bupnormal}.
This work was partially supported by a Grant-in-Aid for Scientific Research of Japan Society for the Promotion of Science ($25 \cdot 3088$, 16H06711); and by the Program for Leading Graduate Schools, MEXT, Japan.

\appendix

\section{Algebraic surfaces}\label{sec:surf}

In this appendix we specify the notation used throughout the paper and we recall some basic results on algebraic surfaces.
We shall not give the proofs of propositions since these can be found in many textbooks such as \cite{hartshorne, beauville, badescu}.

In this paper, a {\em surface} means a smooth projective variety of dimension $2$ over $\mathbb{C}$, which becomes a compact complex manifold of dimension $2$.

\begin{notation}
\hspace{1em}
\begin{itemize}
\item
$\sim$: the linear equivalence of divisors, which is the same as the numerical equivalence if the surface is rational.

\item
$[D]$: the linear equivalence class of $D$.

\item
$\operatorname{Pic} X$: the Picard group of $X$.

\item
$\operatorname{Pic}_{\mathbb{Q}} X = \operatorname{Pic} X \otimes \mathbb{Q}$,
$\operatorname{Pic}_{\mathbb{R}} X = \operatorname{Pic} X \otimes \mathbb{R}$,
$\operatorname{Pic}_{\mathbb{C}} X = \operatorname{Pic} X \otimes \mathbb{C}$.

\item
$\rho(X)$: the Picard number of $X$, which is equal to $\dim_{\mathbb{Q}}(\operatorname{Pic}_\mathbb{Q} X)$ if $X$ is rational.

\item
$H^{i}(X, D) = H^i(D)$: the $i$-th cohomology group of the divisor $D$ (or its class).

\item
$h^{i}(D) = h^{i}(X, D) = \dim_{\mathbb{C}} H^{i}(X, D)$.

\item
$|F|$: the linear system of $F$.

\item
$\dim |F| = h^0(F) - 1$ where $F$ is effective.

\item
$K_X$: the canonical class of $X$.

\item
$D_1 \cdot D_2$: the intersection number of the divisors $D_1$ and $D_2$ (or their classes).

\item
$X \to Y$: a morphism from $X$ to $Y$.

\item
$X \dasharrow Y$: a birational map from $X$ to $Y$.

\item
$\mathcal{O}_{\mathbb{P}^2}(1)$: the class of lines in $\mathbb{P}^2$, which is a generator of $\operatorname{Pic} \mathbb{P}^2$ as a $\mathbb{Z}$-module.

\item
$\operatorname{PGL}(3)$: the set of projective linear transformations on $\mathbb{P}^2$, which coincides with the group of automorphisms on $\mathbb{P}^2$.

\item
$g_a(C) = \dim_{\mathbb{C}} H^1(C, \mathcal{O}_C)$: the arithmetic genus of an irreducible curve $C$.

\end{itemize}
\end{notation}

\begin{theorem}[Hodge index theorem]\label{app:hodge}
The intersection number on a surface $X$ has signature $(1, \rho(X) - 1)$.
\end{theorem}

\begin{definition}
Let $f \colon X \dasharrow Y$ be a birational map and let $\pi \colon \widetilde{X} \to X$ be a resolution of the indeterminacies of $f$:
\[
\xymatrix{
	\widetilde{X} \ar[d]_{\pi} \ar[rd]^{g} & \\
	X \ar@{.>}[r]_{f} & Y.
}
\]
Then, $f_{*}$ and $f^{*}$ are defined by
\[
	f_{*} = g_{*} \pi^{*}, \quad
	f^{*} = \pi_{*} g^{*}.
\]
It is known that these linear maps do not depend on the choice of $\pi$.

These linear maps do not preserve the intersection number in general.
\end{definition}

\begin{proposition}
The above $f_{*}$ and $f^{*}$ preserve the set of effective classes.
\end{proposition}

\begin{remark}
If $f \colon Z \to Y$ and $g \colon Y \to X$ are birational morphisms, then
\[
	(g \circ f)_{*} = g_{*} f_{*} \quad
	\text{ and } \quad
	(g \circ f)^{*} = f^{*} g^{*}.
\]
However, these relations do {\em not} hold if $f, g$ are simply birational maps.
\end{remark}

\begin{definition}[degree]\label{app:p2degree}
Using the homogeneous coordinate $(z_1 : z_2 : z_3) \in \mathbb{P}^2$, a birational automorphism $\varphi$ on $\mathbb{P}^2$ can be written as
\[
	\varphi(z_1 : z_2 : z_3) = (\varphi_1(z_1, z_2, z_3) : \varphi_2(z_1, z_2, z_3) : \varphi_3(z_1, z_2, z_3)),
\]
where $\varphi_1, \varphi_2, \varphi_3$ are homogeneous polynomials of $z_1, z_2, z_3$ of the same degree with no common factors.
The {\em degree} of $\varphi$ (as a birational automorphism on $\mathbb{P}^2$) is defined by $\deg \varphi_i$.

It is known that $\deg \varphi$ can be calculated as follows:
\[
	\deg \varphi = 
	(\varphi_{*}\mathcal{O}_{\mathbb{P}^2}(1)) \cdot \mathcal{O}_{\mathbb{P}^2}(1) = 
	(\varphi^{*}\mathcal{O}_{\mathbb{P}^2}(1)) \cdot \mathcal{O}_{\mathbb{P}^2}(1).
\]
\end{definition}

\begin{definition}[basic rational surface]\label{app:rational}
A rational surface that admits a birational morphism to $\mathbb{P}^2$ is called a {\em basic rational surface}.
This means that a basic rational surface can always be obtained by a finite number of blow-ups of $\mathbb{P}^2$.
\end{definition}

\begin{definition}[geometric basis \cite{dol1}]\label{app:geometricbasis}
Let $X$ be a basic rational surface and $(e^{(0)}, \ldots, e^{(r)})$ be a $\mathbb{Z}$-basis on $\operatorname{Pic} X$.
Then $(e^{(0)}, \ldots, e^{(r)})$ is said to be a {\em geometric basis} if there is a composition of blow-ups $\pi = \pi^{(1)} \circ \cdots \circ \pi^{(r)} \colon X \to \mathbb{P}^2$ such that $e^{(0)} = \pi^{*}\mathcal{O}_{\mathbb{P}^2}(1)$ and $e^{(i)}$ is the total transform of the class of the exceptional curve of $\pi^{(i)}$ for $i = 1, \ldots, r$.

Note that $e^{(i)} - e^{(j)}$ ($i, j > 0, i \ne j$) is effective if and only if $i < j$ and the center of the $j$-th blow-up is infinitely near that of the $i$-th blow-up.

For any geometric basis $(e^{(0)}, \ldots, e^{(r)})$ we have
\[
	K_X = - 3 e^{(0)} + e^{(1)} + \cdots + e^{(r)}, \quad
	e^{(i)} \cdot e^{(j)} = 
	\begin{cases}
		1 & (i = j = 0) \\
		-1 & (i = j \ne 0) \\
		0 & (i \ne j).
	\end{cases}
\]

A birational morphism to $\mathbb{P}^2$ is determined only by its geometric basis up to automorphism on $\mathbb{P}^2$, i.e.\,if two birational morphisms $\pi, \pi' \colon X \to \mathbb{P}^2$ give the same geometric basis on $\operatorname{Pic} X$, then there exists $f \in \operatorname{PGL}(3)$ such that $\pi' = f \circ \pi$.
In fact, these birational morphisms are determined only by $e^{(0)}$, since the set $\{ e^{(1)}, \ldots, e^{(r)} \}$ is determined by
\[
	\left\{
	e^{(1)}, \ldots, e^{(r)} \right\} =
	\left\{ F \in \operatorname{Pic} X \, | \, F^2 = -1, F \cdot e^{(0)} = 0, F \colon \text{effective}
	\right\}.
\]
\end{definition}

\begin{proposition}\label{app:irrednega}
Let $C_1, \ldots, C_m$ be irreducible curves in a surface $X$ such that the matrix $(C_i \cdot C_j)_{ij}$ is negative definite.
Then, for all nonnegative integers $a_1, \ldots, a_m$, we have
\[
	h^0 \left( \sum_i a_i C_i \right) = 1.
\]
In particular, if an irreducible curve $C$ has a negative self-intersection, then
\[
	h^0(m C) = 1
\]
for $m \ge 0$ and thus the class $[C]$ cannot be written as a nontrivial sum of effective classes.
\end{proposition}

\begin{theorem}[Riemann-Roch]\label{app:riemannroch}
Let $F \in \operatorname{Pic} X$.
Then
\[
	h^0(F) - h^1(F) + h^2(F) = \chi(\mathcal{O}_X) + \frac{1}{2} F \cdot (F - K_X),
\]
where we denote by $\chi(\mathcal{O}_X)$ the Euler-Poincar\'{e} characteristic.
Since $h^1(F)$ is not negative, we have
\[
		h^0(F) + h^2(F) \ge \chi(\mathcal{O}_X) + \frac{1}{2} F \cdot (F - K_X),
\]
which is called the {\em Riemann-Roch inequality}.
If $X$ is rational, then $\chi(\mathcal{O}_X) = 1$.
Therefore we have
\[
	h^0(F) - h^1(F) + h^2(F) = 1 + \frac{1}{2} F \cdot (F - K_X)
\]
and
\[
	h^0(F) + h^2(F) \ge 1 + \frac{1}{2} F \cdot (F - K_X).
\]
\end{theorem}

\begin{theorem}[Serre duality]\label{app:serre}
Let $F \in \operatorname{Pic} X$.
Then
\[
	h^i(F) = h^{2-i}(K_X - F)
\]
for $i = 0, 1, 2$.
\end{theorem}

\begin{theorem}[genus formula]\label{app:genus}
Let $C \subset X$ be an irreducible curve.
Then
\[
	g_a(C) = 1 + \frac{1}{2} C \cdot (C + K_X),
\]
where $g_a(C) = h^{1}(C, \mathcal{O}_C)$ is the arithmetic genus of $C$.
\end{theorem}

\begin{proposition}\label{app:p1arith}
Let $C$ be a (possibly singular) irreducible curve.
Then $g_a(C) = 0$ if and only if $C$ is isomorphic to $\mathbb{P}^1$.
\end{proposition}

\begin{definition}[exceptional curve of the first kind]\label{app:exceptionalcurve}
An irreducible curve $C \subset X$ is called an {\em exceptional curve of the first kind} if $C$ is isomorphic to $\mathbb{P}^1$ and $C^2 = -1$.
The genus formula and Proposition~\ref{app:p1arith} imply that these conditions are equivalent to $C^2 = C \cdot K_X = -1$.
\end{definition}

\begin{theorem}[Castelnuovo's contraction theorem]\label{app:castelnuovo}
Let $X$ be a surface and $C \subset X$ an exceptional curve of the first kind.
Then there exist a surface $X'$ and a birational morphism $\pi \colon X \to X'$ such that $\pi(C)$ is a point in $X'$ and $\pi$ is an isomorphism from $X \setminus C$ to $X' \setminus \pi(C)$.
This procedure is called a {\em blow-down}.
In other words, we can contract an exceptional curve of the first kind by a blow-down.
\end{theorem}

\begin{definition}[nef]\label{app:nefclosed}
A divisor $D$ on a surface $X$ is {\em nef} if it satisfies
\[
	C \cdot D \ge 0
\]
for every irreducible curve $C \subset X$.
A class $F \in \operatorname{Pic} X$ (or $F \in \operatorname{Pic}_{\mathbb{R}} X$) is said to be nef if it satisfies
\[
	C \cdot F \ge 0
\]
for every irreducible curve $C \subset X$.

It is known that the self-intersection of a nef class is always nonnegative.
It is also known that $K_X$ cannot be nef for a rational surface $X$.

The set
\[
	\{
		F \in \operatorname{Pic}_{\mathbb{R}} X \, | \, F \colon \text{ nef}
	\}
\]
is a closed convex cone in $\operatorname{Pic}_{\mathbb{R}} X$, i.e.
\begin{itemize}
\item
if $F, F'$ are nef, then so is $F + F'$,
\item
if $F$ is nef and $a > 0$, then $a F$ is also nef,
\item
the above set is a closed set in $\operatorname{Pic}_{\mathbb{R}} X$.
\end{itemize}

The set of all nef classes in $\operatorname{Pic}_{\mathbb{R}} X$ is called the {\em nef cone} of $X$.
\end{definition}

\begin{theorem}[Nagata \cite{nagata}]\label{app:nagata}
If a rational surface has infinitely many exceptional curves of the first kind, then it is a basic rational surface.
\end{theorem}


\section{Proof of Lemma~\ref{lem:symmetricjordan}}\label{sec:proofoflemma}

\begin{proof}[Proof of Lemma~\ref{lem:symmetricjordan}]

Let us extend $(-, -)$ to a Hermitian form on $V_{\mathbb{C}}$.

We show (1) in Steps 1--9.

\noindent
\underline{Step~1}

First let us consider the case where $f$ has an eigenvalue whose modulus is not $1$.

If $\lambda$ is an eigenvalue with $|\lambda| \ne 1$ and $v$ its corresponding eigenvector, then $v$ is always isotropic since
\begin{align*}
	(v, v) &= (f v, f v) \\
	&= |\lambda|^2 (v, v).
\end{align*}

\noindent
\underline{Step~2}

Let us show that if $\lambda$ is an eigenvalue whose modulus is not $1$, then $\lambda$ is simple.

Let $v$ be the corresponding eigenvector and assume $\lambda$ is not simple.
Then there exists $w$, linearly independent of $v$, such that
\[
	f w = \lambda w \quad \text{or} \quad
	f w = \lambda w + v.
\]
In the first case, $v$ and $w$ are orthogonal to each other since
\[
	(v, w) = |\lambda|^2 (v, w),
\]
which, together with $(v, v) = (w, w) = 0$, contradicts Lemma~\ref{lem:lem1}.
Let us consider the second case.
Since
\begin{align*}
	(v, w) &= (\lambda v, \lambda w + v) \\
	&= |\lambda|^2 (v, w),
\end{align*}
we have
\[
	(v, w) = 0.
\]
In the same way we find
\[
	(w, w) = |\lambda|^2 (w, w)
\]
and thus
\[
	(v, v) = (v, w) = (w, w) = 0,
\]
which again contradicts Lemma~\ref{lem:lem1}.

\noindent
\underline{Step~3}

We show that if $\lambda_1, \lambda_2$ are different eigenvalues of $f$ with $|\lambda_i| \ne 1$, then $\lambda_2 = 1/\overline{\lambda_1}$.
In particular, $\lambda_i$ must be real numbers.

Let $v_1, v_2$ be the corresponding eigenvectors.
These vectors are both isotropic by Step~1, and linearly independent since they are eigenvectors corresponding to different eigenvalues.
Thus it follows from Lemma~\ref{lem:lem1} that $(v_1, v_2) \ne 0$.
Since
\[
	(v_1, v_2) = \lambda_1 \overline{\lambda_2} (v_1, v_2),
\]
we have
\[
	\lambda_1 \overline{\lambda_2} = 1.
\]

\noindent
\underline{Step~4}

We show that if $f$ has an eigenvalue whose modulus is not $1$, then $f$ is diagonalizable.
We already know that such eigenvalues are simple.

We therefore consider an eigenvalue $\mu$ of modulus $1$ and assume vectors $u_1, u_2$ satisfy
\begin{align*}
	f u_1 &= \mu u_1, \\
	f u_2 &= \mu u_2 + u_1.
\end{align*}
Since
\[
	(u_1, u_2) = (u_1, u_2) + \mu(u_1, u_1),
\]
$u_1$ is isotropic.
In the same way we find
\[
	(v, u_1) = \lambda \overline{\mu}(v, u_1)
\]
and thus
\[
	(v, v) = (v, u_1) = (u_1, u_1) = 0,
\]
which contradicts Lemma~\ref{lem:lem1}.
Hence, $f$ is diagonalizable.

\noindent
\underline{Step~5}

We show that if $\lambda \in \mathbb{R} \setminus \{ \pm 1 \}$ is an eigenvalue of $f$, then so is $1/\lambda$.

Assume that $1/\lambda$ is not an eigenvalue.
Then, since all eigenvalues except $\lambda$ have modulus $1$ (Step~3) and $f$ is diagonalizable (Step~4), there exists a basis $(v, u_1, \ldots, u_r)$ of $V_{\mathbb{C}}$ such that
\begin{align*}
	f v &= \lambda v, \\
	f u_i &= \mu_i u_i,
\end{align*}
where $|\mu_i| = 1$.
Since
\[
	(v, u_i) = \lambda \overline{\mu_i}(v, u_i),
\]
we have $(v, u_i) = 0$.
Therefore, using Step~1, we obtain that $v$ is orthogonal to all elements in $V_{\mathbb{C}}$.
However, this contradicts the nondegeneratedness of $(-, -)$.

\hspace{1em}

Steps 1--5 show that if $f$ has an eigenvalue whose modulus is not $1$, then the Jordan normal form of $f$ is (\ref{eq:jordanexp}).
From now on, let us consider the case where all eigenvalues of $f$ have modulus $1$.

\noindent
\underline{Step~6}

It is clear that if $f$ is diagonalizable, then its Jordan normal form is (\ref{eq:jordanbounded}).
Thus it is sufficient to show that if $f$ is not diagonalizable, then its Jordan normal form is (\ref{eq:jordanquadratic}).

\noindent
\underline{Step~7}

We show that the size of each Jordan block is at most $3$.

Assume that linearly independent vectors $v_1, v_2, v_3, v_4$ satisfy
\begin{align*}
	f v_1 &= \nu v_1, \\
	f v_2 &= \nu v_2 + v_1, \\
	f v_3 &= \nu v_3 + v_2, \\
	f v_4 &= \nu v_4 + v_3.
\end{align*}
Using
\[
	(v_1, v_i) = (v_1, v_i) + \nu (v_1, v_{i-1}),
\]
we have $(v_1, v_{i-1}) = 0$ for $i = 2, 3, 4$.
Since
\begin{align*}
	(v_2, v_3) &= (\nu v_2 + v_1, \nu v_3 + v_2) \\
	&= (v_2, v_3) + \nu (v_2, v_2),
\end{align*}
$v_2$ is isotropic, which contradicts Lemma~\ref{lem:lem1}.

\noindent
\underline{Step~8}

We show that $f$ has only one Jordan block whose size is greater than $1$.
In particular, the corresponding eigenvalue is $\pm 1$.

Let $\mu, \nu$ be eigenvalues of modulus $1$ and let pairwise-linearly independent vectors $v_1, v_2$ and $w_1, w_2$ satisfy
\begin{align*}
	f v_1 & = \mu v_1,
	& f w_1 & = \nu w_1, \\
	f v_2 & = \mu v_2 + v_1,
	& f w_2 & = \nu w_2 + w_1.
\end{align*}
It is sufficient to show that $v_1$ and $w_1$ are linearly dependent.

The same calculation as in Step~4 implies that $v_1$ and $w_1$ are both isotropic.
Therefore, since
\[
	(v_1, w_1) = \mu \overline{\nu} (v_1, w_1),
\]
it follows from Lemma~\ref{lem:lem1} that $\mu = \nu$.
However, using
\[
	(v_1, w_2) = (v_1, w_2) + \mu (v_1, w_1),
\]
we have $(v_1, w_1) = 0$.
Hence Lemma~\ref{lem:lem1} shows that $v_1$ and $w_1$ are linearly dependent.

\noindent
\underline{Step~9}

Finally we show that the size of the Jordan block in Step~8 is exactly $3$.
It is sufficient to show that the size is not $2$.

Assume that $(v_1, v_2, u_1, \ldots, u_{r-1})$ is a basis on $V_{\mathbb{C}}$ such that
\begin{align*}
	f v_1 &= \nu v_1, \\
	f v_2 &= \nu v_2 + v_1, \\
	f u_i &= \mu_i u_i.
\end{align*}
Moreover, we can take $v_1$ and $v_2$ in $V$ since $\nu = \pm 1$.
As in Step~5, we deduce a contradiction by showing that $v_1$ is orthogonal to any of $v_1, v_2, u_1, \ldots, u_{r-1}$, bearing in mind that a calculation similar to Step~4 shows that $(v_1, v_1) = 0$.
Using
\[
	(v_2, v_2) = (v_2, v_2) + 2 \nu (v_1, v_2),
\]
we have $(v_1, v_2) = 0$.
Finally we show that $(v_1, u_i) = 0$.
If $\nu \ne \mu_i$, then
\[
	(v_1, u_i) = \nu \overline{\mu_i} (v_1, u_i)
\]
and thus $(v_1, u_i) = 0$.
If $\nu = \mu_i$, then
\[
	(v_2, u_i) = (v_2, u_i) + \nu (v_1, u_i)
\]
and thus $(v_1, u_i) = 0$.

\hspace{1em}

We now show (2).

\noindent
\underline{Step~10}

Let us represent $w$ as
\[
	w = a_1 v_1 + a_2 v_2 + a_3 v_3 + \sum_j b_j u_j.
\]
Since $(v_1, \ldots, u_{r-2})$ is a Jordan basis corresponding to (\ref{eq:jordanquadratic}), we have
\[
	f^n w = \left( \nu^n a_1 + \nu^{n-1} n a_2 + \frac{n(n-1)\nu^n}{2} a_3 \right) v_1 + (\nu^n a_2 + \nu^{n-1} n a_3) v_2 + \nu^n a_3 v_3 + \sum_j \mu^n_j b_j u_j.
\]
Dividing both sides by $\nu^n n^2$ and taking the limit $n \to +\infty$, we have
\[
	\lim_{n \to +\infty} \frac{1}{\nu^n n^2} f^n w = \frac{a_3}{2} v_1.
\]
Thus it is sufficient to show that $a_3 = {(w, v_1)}/{(v_3, v_1)}$.
A calculation similar to that given in Step~9 leads to
\[
	(v_1, v_1) = (v_2, v_1) = (u_j, v_1) = 0.
\]
Therefore, we have
\[
	(w, v_1) = a_3 (v_3, v_1)
\]
and thus
\[
	a_3 = \frac{(w, v_1)}{(v_3, v_1)}.
\]

\hspace{1em}

We finally show (3) in a similar way as in Step~10.

\noindent
\underline{Step~11}

Let us represent $w$ as
\[
	w = a_1 v_1 + a_2 v_2 + \sum_j b_j u_j.
\]
Then we have
\[
	f^n w = \lambda^n a_1 v_1 + \frac{a_2}{\lambda^n} v_2 + \sum_j \mu^n_j b_j u_j
\]
and
\[
	\lim_{n \to +\infty} \frac{1}{\lambda^n} w = a_1 v_1.
\]
Step~1 and a calculation similar to that given in Step~5 lead to
\[
	(v_1, v_1) = (v_1, u_j) = 0
\]
and thus we have
\[
	a_1 = \frac{(w, v_2)}{(v_1, v_2)}.
\]
\end{proof}


\end{document}